\documentclass[final]{siamltex}
\usepackage{epsfig}
\usepackage{graphicx}
\usepackage{epstopdf}
\usepackage{amsmath}
\usepackage{amssymb}
\usepackage{setspace}
\usepackage{enumerate}
\usepackage{graphicx}
\usepackage{array}
\newcolumntype{C}[1]{>{\centering\let\newline\\\arraybackslash\hspace{0pt}}m{#1}}
\usepackage[ruled,linesnumbered,vlined]{algorithm2e}
\usepackage{subfigure}
\usepackage{bm}
\usepackage{microtype}
\usepackage{hyperref}

\newcommand{\R}{\mathbb{R}}
\newcommand{\p}{\partial}
\renewcommand{\vec}[1]{\mathbf{#1}}

\usepackage{hyperref,color}

\definecolor{webgreen}{rgb}{0,.35,0}
\definecolor{webbrown}{rgb}{.6,0,0}
\definecolor{RoyalBlue}{rgb}{0,0,0.9}
\definecolor{purp}{rgb}{0.6,0.05,0.8}
\definecolor{ora}{rgb}{0.7,0.35,0.02}

\hypersetup{
   colorlinks=true, linktocpage=true, pdfstartpage=3, pdfstartview=FitV,
   breaklinks=true, pdfpagemode=UseNone, pageanchor=true, pdfpagemode=UseOutlines,
   plainpages=false, bookmarksnumbered, bookmarksopen=true, bookmarksopenlevel=1,
   hypertexnames=true, pdfhighlight=/O,
   urlcolor=webbrown, linkcolor=RoyalBlue, citecolor=webgreen,
   pdfauthor={Gary P. T. Choi and Chris H. Rycroft},
   pdfsubject={Density-Equalizing Maps for Simply-Connected Open Surfaces},
   pdfkeywords={},
   pdfcreator={pdfLaTeX},
   pdfproducer={LaTeX with hyperref}
}

\title{Density-Equalizing Maps for Simply-Connected Open Surfaces}
\author{Gary P.~T.~Choi\footnotemark[1] \and Chris H.~Rycroft\footnotemark[1]\ \footnotemark[2]}

\begin{document}

\maketitle
\renewcommand{\thefootnote}{\fnsymbol{footnote}}
\footnotetext[1]{Paulson School of Engineering and Applied Sciences, Harvard University, MA 02138.}
\footnotetext[2]{Department of Mathematics, Lawrence Berkeley Laboratory, Berkeley, CA 94720.}
\renewcommand{\thefootnote}{\arabic{footnote}}

\begin{abstract}
In this paper, we are concerned with the problem of creating flattening maps of simply-connected open surfaces in $\R^3$. Using a natural principle of density diffusion in physics, we propose an effective algorithm for computing density-equalizing flattening maps with any prescribed density distribution. By varying the initial density distribution, a large variety of mappings with different properties can be achieved. For instance, area-preserving parameterizations of simply-connected open surfaces can be easily computed. Experimental results are presented to demonstrate the effectiveness of our proposed method. Applications to data visualization and surface remeshing are explored.
\end{abstract}


\pagestyle{myheadings}
\thispagestyle{plain}
\markboth{Gary P. T. Choi and Chris H. Rycroft}{Density-Equalizing Maps for Simply-Connected Open Surfaces}


\section{Introduction}
The problem of producing maps has been tackled by scientists and cartographers for centuries. A classical map-making problem is to flatten the globe onto a plane. Numerous methods have been proposed, each aiming to preserve different geometric quantities. For instance, the Mercator projection produces a conformal planar map of the globe: angles and small objects are preserved but the area near the poles is seriously distorted.

One problem in computer graphics closely related to cartogram production is surface parameterization, which refers to the process of mapping a complicated surface to a simpler domain. With the advancement of the computer technology, three-dimensional (3D) graphics have become widespread in recent decades. To create realistic textures on 3D shapes, one common approach is to parameterize the 3D shapes onto $\R^2$. The texture can be designed on $\R^2$ and then be mapped back onto the 3D shapes. Again, different criteria of distortion minimization have led to the invention of a large number of parameterization algorithms.

Gastner and Newman~\cite{Gastner04} proposed an algorithm for producing density-equalizing cartograms based on the diffusion equation. Specifically, given a map and certain data defined on each part of the map (such as the population at different regions), the algorithm deforms the map such that the density, defined by the population per unit area, becomes a constant all over the deformed map. The diffusion-based cartogram generation approach has been widely used for data visualization. For instance, Dorling~\cite{Dorling10} applied this approach to visualize sociological data such as the global population, the income and the age-of-death at different regions. Colizza et al.~\cite{Colizza06} constructed a geographical representation of disease evolution in the United States for an epidemics using this cartogram generation algorithm. Wake and Vredenburg~\cite{Wake08} visualized global amphibian species diversity using the method. Other applications include the visualization of the democracies and autocracies of different countries~\cite{Gleditsch06}, the race/ethnicity distribution of Twitter users in the United States~\cite{Mislove11}, the rate of obesity for individuals in Canada~\cite{Vanasse06}, and the world citation network~\cite{Pan12}.

Inspired by the above approach, we develop an efficient finite-element algorithm for computing density-equalizing flattening maps of simply-connected open surfaces in $\R^3$ onto $\R^2$. Given a simply-connected open triangulated surface and certain quantities defined on all triangle elements of the surface, we first flatten the surface onto $\R^2$ by a natural flattening map. Then, the flattened surface is deformed according to the given quantities using a fast iterative scheme. Furthermore, by altering the input quantities defined on the triangle elements, flattening maps with different properties can be achieved. For instance, area-preserving parameterizations of simply-connected open surfaces can be easily obtained. 

\subsection{Contribution}
The contribution of our work for computing density-equalizing flattening maps of simply-connected open surfaces is as follows.
\begin{enumerate}[(i)]
 \item Our approach is applicable to a wider class of surfaces when compared to the previous approach by Gastner and Newman~\cite{Gastner04}. The previous approach~\cite{Gastner04} works for two-dimensional (2D) domains while ours works for simply-connected open surfaces in $\R^3$.
 \item We propose a linear formulation for computing a boundary-aware convex flattening map of simply-connected open surfaces. The flattening map effectively preserves the curvature of the input surface boundary and serves as a good initial mapping for the subsequent density-equalizing process.
 \item We propose a new scheme for constructing an auxiliary region for the density diffusion. When compared to the previous approach~\cite{Gastner04}, which makes use of a regular rectangular grid for constructing the auxiliary region, our approach produces a more adaptive auxiliary region that requires fewer points and hence reduces the computational cost. 
 \item We propose a finite-element iterative scheme for solving the density-diffusion problem without introducing the Fourier space as in the previous approach~\cite{Gastner04}. The scheme accelerates the computation for density-equalizing maps, with the accuracy well-preserved.
 \item Our proposed algorithm can be used for a wide range of applications, including the computation of area-preserving parameterizations, data visualization, and surface remeshing.
\end{enumerate}

\subsection{Organization of the paper}
In Section \ref{sect:previous}, we review the previous works on cartogram generation and surface parameterization. The physical principle of density-equalization is outlined in Section \ref{sect:background}. In Section \ref{sect:main}, we describe our proposed method for achieving density-equalizing flattening maps of simply-connected open surfaces. Experimental results are presented in Section \ref{sect:experiment} for analyzing our proposed algorithm. In Section \ref{sect:applications}, we discuss two applications of our algorithm. In Section \ref{sect:discussion}, we conclude this paper with a discussion on the limitation of our current approach and possible future works. 

\section{Previous work} \label{sect:previous}
The problem of map generation has been studied by cartographers, geographers and scientists for centuries. Readers are referred to Dorling~\cite{Dorling96} for a short survey of pre-existing cartogram production methods. Edelsbrunner and Waupotitsch~\cite{Edelsbrunner97} proposed a combinatorial approach to construct homeomorphisms with prescribed area distortion for cartogram generation. Keim et al.~\cite{Keim04} developed the Cartodraw algorithm for producing contiguous cartograms. Gastner and Newman~\cite{Gastner04} proposed a cartogram production algorithm based on density diffusion. Keim et al.~\cite{Keim05} proposed to use medial-axis-based transformations for making cartograms. 

In this work, cartogram generation is shown to be closely related to surface parameterization. For surface parameterization, a large variety of algorithms have been proposed by different research groups. There are two major classes of surface parameterization algorithms, namely conformal parameterization and authalic parameterization. Conformal parameterization aims to preserve the angles and hence the infinitesimal shapes of the surfaces while sacrificing the area ratios. By contrast, authalic parameterization aims to preserve the area measure of the surfaces while neglecting angular distortions. For the existing parameterization algorithms, we refer the readers to the comprehensive surveys~\cite{Floater05,Sheffer06,Hormann07}. We highlight the works on surface parameterization in the recent decade. For conformal parameterization, recent advances include the discrete Ricci flow method~\cite{Jin08,Yang09,Zhang14} and the quasi-conformal composition~\cite{Choi15a,Choi15b,Choi16,Meng16,Choi17}. For area-preserving parameterization, the state-of-the-art approaches are mainly based on Lie advection of differential 2-forms~\cite{Zou11} and optimal mass transport~\cite{Zhao13,Su16}. Recently, Nadeem et al.~\cite{Nadeem16} proposed an algorithm for achieving spherical parameterization with controllable area distortion.

\section{Background} \label{sect:background}
Our work aims to produce flattening maps based on a physical principle of diffusion. The diffusion-based method for producing cartogram proposed by Gastner and Newman~\cite{Gastner04} is outlined as follows. For the rest of the paper, we denote the method by Gastner and Newman~\cite{Gastner04} as \textit{GN}. Given a planar map and a quantity called the \textit{population} defined on every part of the map, let $\rho$ be the density field defined by the quantity per unit area. The map can be deformed by equalizing the density field $\rho$ using the advection equation
\begin{equation}
\frac{\p \rho}{\p t} = - \nabla \cdot \vec{j}
\end{equation}
where the flux is given by Fick's law,
\begin{equation}
\vec{j} = - \nabla \rho.
\end{equation}
This yields the diffusion equation
\begin{equation}\label{eqt:diffusion_continuous}
\frac{\p \rho}{\p t} = \Delta \rho. 
\end{equation}
Since time can be rescaled in the subsequent analysis, the diffusion constant in Fick's law is set to 1. Any tracers that are being carried by this density flux will move with velocity
\begin{equation} \label{eqt:velocity}
  \vec{v}(\vec{r},t) = \frac{\vec{j}}{\rho} = - \frac{\nabla \rho}{\rho}.
\end{equation}

If \eqref{eqt:diffusion_continuous} is solved to steady state, and the map is deformed according to the velocity field in \eqref{eqt:velocity}, then the final state of the map will have equalized density. To track the deformation of the map, Gastner and Newman introduce tracers $\vec{r}(t)$ that follow the velocity field according to
\begin{equation}
\vec{r}(t) = \vec{r}(0) + \int_0^t \vec{v}(\vec{r},\tau) d \tau.
\end{equation}
In other words, taking $t \to \infty$, the above displacement $\vec{r}(t)$ produces a map that achieves equalized density per unit area. To avoid infinite expansion of the map, GN proposed to construct a large rectangular auxiliary region, called the \textit{sea}, surrounding the region of interest. By defining the density at the sea to be the average density of the region of interest, it can be ensured that the area of the deformed map is as same as that of the initial map. In GN, the above procedures were developed using finite difference grids and the above equations were solved in Fourier space. 

There is room for improvement of the abovementioned approach in two major aspects. First, the above 2D finite difference approach works for planar domains but not for general simply-connected open surfaces in $\R^3$. Second, the large rectangular sea and the large number of grid points may cause long computational time. In this work, an algorithm that further enhances the abovementioned approach in the two aspects is proposed. The details of our proposed algorithm in described in the following section. 

\section{Our proposed method} \label{sect:main}
We now describe our method for computing density-equalizing maps of simply-connected open surfaces in $\R^3$. Let $S$ be a simply-connected open surface in $\R^3$ and $\rho$ be a prescribed density distribution. Our goal is to compute a flattening map $f:S \to \R^2$ such that the Jacobian $J_f$ satisfies 
\begin{equation}
J_f \propto \rho.
\end{equation}
In other words, the final density per unit area in the flattening map becomes a constant.

Our proposed algorithm primarily consists of three steps, described in Secs.~\ref{sect:step1}, \ref{sect:step2}, and \ref{sect:step3}. We remark that if the input surface is planar, the first step can be skipped. In the following discussions, $S$ is discretized as a triangular mesh $(\mathcal{V}, \mathcal{E}, \mathcal{F})$ where $\mathcal{V}$ is the vertex set, $\mathcal{E}$ is the edge set and $\mathcal{F}$ is the triangular face set. $\rho$ is discretized as $\rho^{\mathcal{F}}$ on every triangle element $T \subset \mathcal{F}$.

\subsection{Initialization: Fast curvature-based flattening map} \label{sect:step1}
To compute the density-equalization process, the first step is to flatten $S$ onto $\R^2$. To minimize the discrepancy between the surface and the flattening result, it is desirable that the outline of the flattening result is similar to the surface boundary. We first simplify the problem by considering only the curve flattening problem of the surface boundary. Then, we construct a surface flattening map based on the curve flattening result.

\subsubsection{Curvature-based flattening of the surface boundary}
Let $\gamma$ be the boundary of the given surface $S$. Note that $\gamma$ is a simple closed curve in $\R^3$ and hence we can write it as an arc-length parameterized curve $\gamma = \gamma(t):[0,l_{\gamma}] \to \R^3$, where $l_{\gamma}$ is the total arc length of $\gamma$. Our goal is to flatten $\gamma$ onto $\R^2$ using a map $\varphi: [0,l_{\gamma}] \to \R^2$ and then obtain the entire flattening map of the surface $S$. For $\gamma$, we can compute two quantities: the curvature $\kappa_{\gamma}$ and the torsion $\tau_{\gamma}$. Note that the curvature $\kappa_{\gamma}$ measures the deviation of $\gamma$ from a straight line, and the torsion $\tau_{\gamma}$ measures the deviation of $\gamma$ from a planar curve. We have the following important theorem.

\begin{theorem}[Fundamental theorem of space curves~\cite{docarmo76}] 
The curve $\gamma$ is completely determined (up to rigid motion) by its curvature $\kappa_{\gamma}$ and its torsion $\tau_{\gamma}$. 
\end{theorem}

Motivated by the above theorem, we consider mapping $\gamma$ to $\varphi(\gamma)$ such that $\kappa_{\gamma} \approx \kappa_{\varphi(\gamma)}$ and $\tau_{\varphi(\gamma)} = 0$. In other words, we project $\gamma$ onto the space of planar convex curves such that the curvature is preserved as much as possible. 
By Frenet--Serret formulas~\cite{docarmo76}, 
\begin{equation}
\vec{T}'(s) = \kappa_{\gamma}(s) \|\gamma'(s)\| \vec{N}(s)
\end{equation}
where $\vec{T}$ and $\vec{N}$ are respectively the unit tangent and unit normal of $\gamma$. It follows that
\begin{equation}
\kappa_{\gamma}(s) = \frac{\|\vec{T}'(s)\|}{\|\gamma'(s)\|}.
\end{equation}
After obtaining $\kappa_{\gamma}$, our goal is to construct a projection of $\gamma$ onto the space of planar simple convex closed curves. Note that for any simple closed planar curve $\mathcal{C} \subset \R^2$, the total signed curvature of $\mathcal{C}$ is a constant~\cite{docarmo76}:
\begin{equation}
\int_{\mathcal{C}} k_{\mathcal{C}}(s) ds = 2 \pi.
\end{equation}
Now, to construct a closed planar curve $\varphi$ with total arclength same as $\gamma$, we set the target signed curvature $k$ to be 
\begin{equation} \label{eqt:nonnegative_curvature_gamma}
k(s) = \frac{ 2 \pi \kappa_{\gamma}(s)}{\int_\gamma \kappa_{\gamma}(t) dt} \geq 0.
\end{equation}
Then, consider the curve
\begin{equation}
\varphi(s) = \left(\int_0^s \cos \theta(u) du, \int_0^s \sin \theta(u) du\right),
\end{equation}
where 
\begin{equation}
\theta(u) = \int_0^u k(t) dt.
\end{equation}
It is easy to check that
\begin{equation}
\varphi'(s) = (\cos \theta(s), \sin \theta(s))
\end{equation}
and hence $\varphi$ is an arclength parameterized curve. Moreover, we have
\begin{equation}
k_{\varphi}(s) = \theta'(s) = k(s) \geq 0.
\end{equation}
However, it should be noted that $\varphi$ may be a closed curve. In other words, there may be a small gap between $\varphi(0)$ and $\varphi(l_{\gamma})$ with $0 \leq \|\varphi(l_{\gamma})-\varphi(0)\| \ll L$, where $L$ is the total arclength of $\varphi$. To enforce that $\varphi$ is closed, we consider updating it by
\begin{equation}
\varphi(s) \leftarrow \varphi(s) - \frac{s}{l_{\gamma}}\left(\varphi(l_{\gamma}) - \varphi(0)\right).
\end{equation}
In fact, $\varphi$ becomes a simple closed convex plane curve under this adjustment. The proof is provided in the appendix.


The algorithm is summarized in Algorithm \ref{alg:embed_curvature}. After obtaining the simple closed convex plane curve $\varphi$, we can use it as a convex boundary constraint and compute a map $\phi: S \to \R^2$ as an initial flattening map of the entire surface $S$. Two methods for surface flattening are suggested below.

\begin{algorithm}[h!]
\label{alg:embed_curvature}
\KwIn{The boundary $\gamma$ of a simply-connected open surface $S$ in $\R^3$.}
\KwOut{A curvature-based flattened curve $\varphi$.}
\BlankLine
Let $\gamma = \{v_j\}_{j=1}^b$ be the boundary vertices of $S$ in anti-clockwise order. Compute the curvature $\kappa =  \frac{\|\vec{T}'\|}{\|\gamma'\|}$\;
Rescale $\kappa$ by $\kappa \leftarrow \frac{ 2 \pi \kappa}{\int_\gamma \kappa(s) ds}$\;

Obtain the flattened curve
$\varphi(s) = \left(\int_0^s \cos \theta(u) du, \int_0^s \sin \theta(u) du\right)$, where $\theta(u) = \int_0^u \kappa(t) dt$\;

Adjust the map by $\varphi(s) \leftarrow \varphi(s) - \frac{s}{l_{\gamma}}\left(\varphi(l_{\gamma}) - \varphi(0)\right)$\;
\caption{Curvature-based curve flattening}
\end{algorithm}

\subsubsection{Curvature-based Tutte flattening map}
One way to construct a bijective planar map $\phi$ is the graph embedding method of Tutte~\cite{Tutte63}. To give an overview of the method, we first introduce the concept of adjacency matrix. The \textit{adjacency matrix} $M$ is a $|\mathcal{V}| \times |\mathcal{V}|$ matrix defined by
\begin{equation}
M_{ij} = \left\{\begin{array}{ll}
1 & \text{ if } [i,j] \in \mathcal{E},\\
0 & \text{ otherwise}.
\end{array}\right.
\end{equation}
In other words, the adjacency matrix only takes the combinatorial information of the input triangular mesh into account and neglects the geometry of it. It was proved by Tutte~\cite{Tutte63} that there exists a bijective map $\phi$ between any simply-connected open triangulated surface $S$ in $\R^3$ and any convex polygon $P$ on $\mathbb{C}$ with the aid of the adjacency matrix. More explicitly, by representing $\phi$ as a complex column vector with length $|\mathcal{V}|$, $\phi$ can be obtained by solving the complex linear system
\begin{equation}
\left\{ \begin{array}{ll}
M^{\text{Tutte}} \phi(v) = 0 & \text{ if } v \in S \setminus \partial S, \\
\phi(\partial S) = \partial P,
\end{array}\right.
\end{equation}
where
\begin{equation}
M^{\text{Tutte}}_{ij} = \left\{\begin{array}{ll}
M_{ij} & \text{ if } [x_i,x_j] \in \mathcal{E},\\
- \sum_{t \neq i} M_{it} & \text{ if } j = i,\\
0 & \text{ otherwise.}
\end{array}\right.
\end{equation}
Here the boundary mapping $\phi: \p S \to \p P$ can be any bijective map. Using our curvature-based flattened curve $\varphi$ as the convex boundary constraint, a bijective Tutte flattening map $\phi: S \to P$ can be easily obtained as described in Algorithm \ref{alg:tutte_curvature}.

\begin{algorithm}[h!]
\label{alg:tutte_curvature}
\KwIn{A simply-connected open surface $S$ in $\R^3$.}
\KwOut{A curvature-based flattening map $\phi: S \to \R^2$.}
\BlankLine
Let $\gamma = \{v_j\}_{j=1}^b$ be the boundary vertices of $S$. Compute the curvature-based curve flattening $\varphi: \gamma \to \mathbb{C}$\;

Compute the adjacency matrix $M$ with $M_{ij} = \left\{\begin{array}{ll}
1 & \text{ if } [i,j] \in \mathcal{E}\\
0 & \text{ otherwise}
\end{array}\right.$\;

Solve the linear system $\left\{ \begin{array}{ll}
M^{\text{Tutte}} \phi(v) = 0 & \text{ if } v \in S \setminus \p S, \\
\phi(v_j) = \varphi(v_j) & \text{ for all } \{v_j\}_{j=1}^b,
\end{array}\right.$ and obtain the desired map $\phi$\;

\caption{Curvature-based Tutte flattening map}
\end{algorithm}

\subsubsection{Curvature-based locally authalic flattening map}

\begin{figure}[t!]
 \centering
 \includegraphics[width=0.3\textwidth]{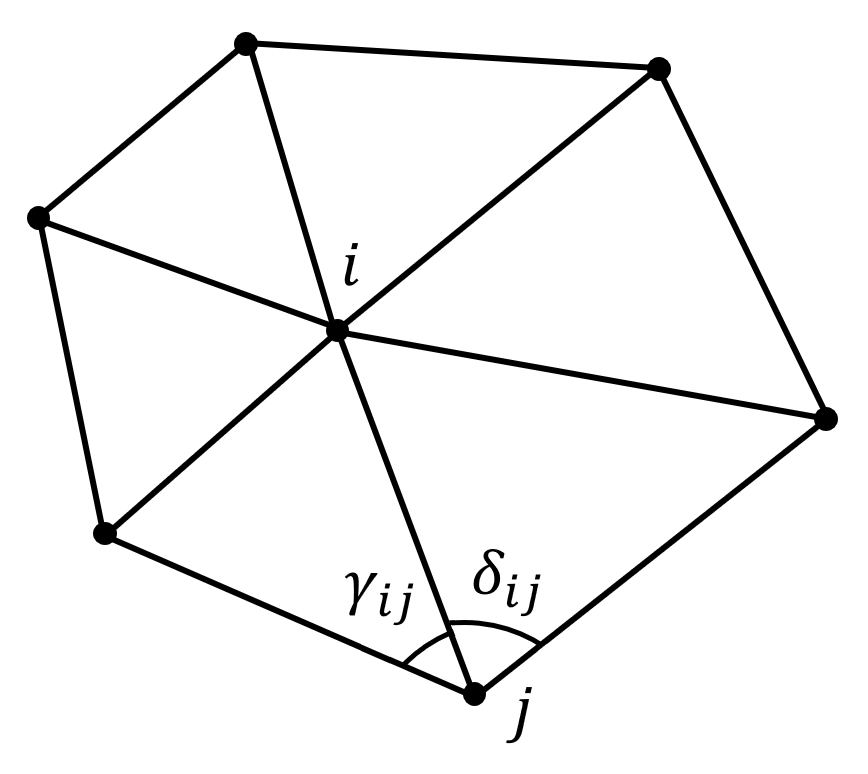}
 \caption{The angles $\gamma_{ij}$ and $\delta_{ij}$ in the locally authalic Chi energy.}
 \label{fig:authalic}
\end{figure}

By changing the matrix $M^{\text{Tutte}}$ in the above method, another way to construct $\phi$ can be obtained. Desbrun et al.~\cite{Desbrun02} proposed a mapping scheme by minimizing the quadratic Chi energy
\begin{equation}
E_{\chi}(\phi) = \sum_{j \in N(i)} \frac{\cot \gamma_{ij} + \cot \delta_{ij}}{|x_i - x_j|^2} |\phi(x_i) - \phi(x_j)|^2,
\end{equation}
where $\gamma_{ij}$ and $\delta_{ij}$ are the two angles at $x_j$ as illustrated in Figure \ref{fig:authalic}. The minimization of the Chi energy aims to find a locally authalic mapping $\phi:S \to \R^2$ that preserves the local 1-ring area at every vertex as much as possible. The associated authalic matrix of this energy is given by
\begin{equation}
M^{\chi}_{ij} = \left\{\begin{array}{ll}
\frac{\cot \gamma_{ij} + \cot \delta_{ij}}{|x_i - x_j|^2} & \text{ if } [x_i,x_j] \in \mathcal{E},\\
- \sum_{t \neq i} M^{\chi}_{it} & \text{ if } j = i,\\
0 & \text{ otherwise.}
\end{array}\right.
\end{equation}

Now consider replacing $M^{\text{Tutte}}$ in the Tutte flattening algorithm by $M^{\chi}$ and solve for a new flattening map. It is noteworthy that the minimizer of the Chi energy is not a globally optimal area-preserving mapping. Nevertheless, it serves as a reasonably good and simple initialization for our density-equalization problem. More explicitly, using our curvature-based boundary constraint, $\phi$ can be obtained by solving the following complex linear system
\begin{equation}
\left\{ \begin{array}{ll}
M^{\chi} \phi(v) = 0 & \text{ if } v \in S \setminus \p S, \\
\phi(\p S) = \varphi.
\end{array}\right.
\end{equation}
The curvature-based locally authalic flattening map is summarized in Algorithm \ref{alg:chi_curvature}.

\begin{algorithm}[h!]
\label{alg:chi_curvature}
\KwIn{A simply-connected open surface $S$ in $\R^3$.}
\KwOut{A curvature-based locally authalic flattening map $\phi: S \to \R^2$.}
\BlankLine
Let $\gamma = \{v_j\}_{j=1}^b$ be the boundary vertices of $S$. Compute the curvature-based curve flattening $\varphi:\gamma \to \mathbb{C}$\;

Compute the authalic matrix 
$M^{\chi}_{ij} = \left\{\begin{array}{ll}
\frac{\cot \gamma_{ij} + \cot \delta_{ij}}{|x_i - x_j|^2} & \text{ if } [x_i,x_j] \in \mathcal{E},\\
- \sum_{t \neq i} M^{\chi}_{it} & \text{ if } j = i,\\
0 & \text{ otherwise}
\end{array}\right.$\;

Solve the linear system $\left\{ \begin{array}{ll}
M^{\chi} \phi(v) = 0 & \text{ if } v \in S \setminus \p S, \\
\phi(v_j) = \varphi(v_j) & \text{ for all } \{v_j\}_{j=1}^b,
\end{array}\right.$ and obtain the desired map $\phi$\;

\caption{Curvature-based locally authalic flattening map}
\end{algorithm}

We remark that both of the two curvature-based flattening algorithms above are a good choice of initialization for our problem for the following reasons:
\begin{enumerate}[(i)]
\item The boundary is flattened as a convex closed planar curve. The convex boundary constraint leads to a bijective flattening map of the surface using our proposed algorithm.
\item The computation of the flattening maps is highly efficient. Both algorithms only involve solving one complex linear system without any iterative procedures.
\item Unlike other conventional parameterizations such as conformal parameterizations, our curvature-based flattening maps result in a relatively uniform distribution of vertices on $\mathbb{C}$ and avoid shrinking particular regions. The diffusion process can then be more accurately executed. 
\end{enumerate}

\subsection{Construction of sea via reflection} \label{sect:step2}

\begin{figure}[t!]
 \centering
 \includegraphics[width=0.4\textwidth]{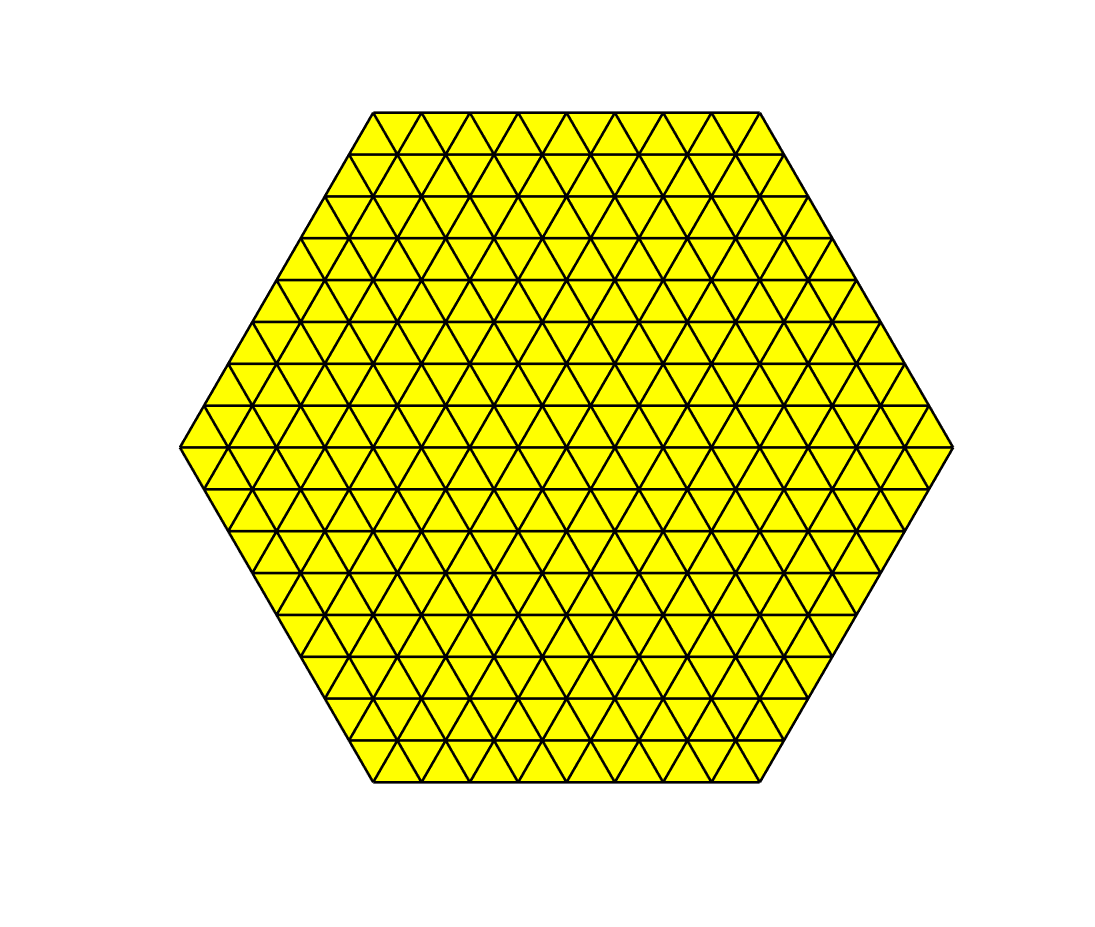}
 \includegraphics[width=0.35\textwidth]{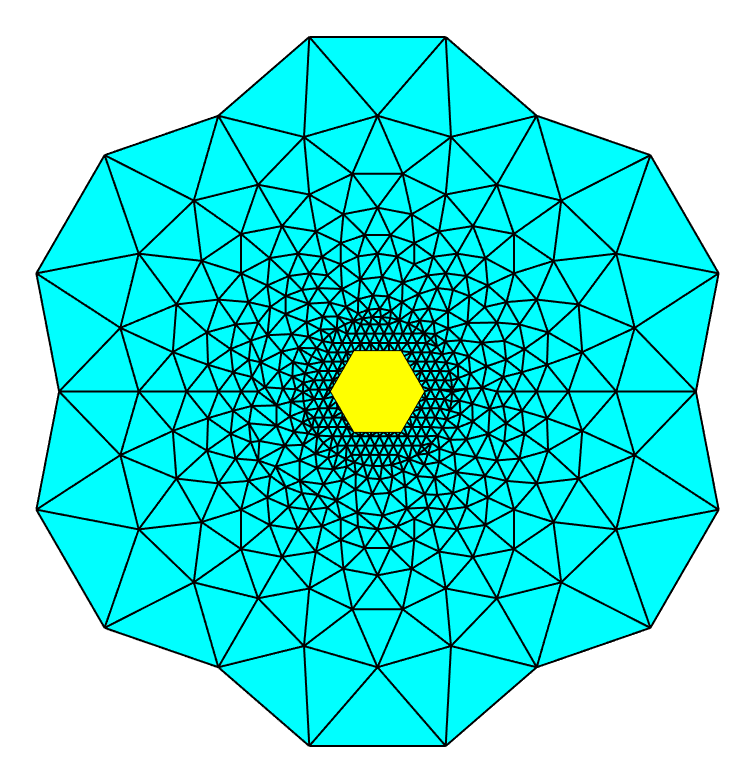}
 \caption{Illustration of our algorithm for constructing the sea. Left: the initial flattening map. We put it inside the unit circle and fill up the gap with uniformly distributed points, and then reflect the entire region along the circle to construct the sea. Right: the sea constructed (in cyan) and the initial flattening map (in yellow).}
 \label{fig:sea}
\end{figure}

In the diffusion-based approach of GN, one important step is to set up a sea surrounding the area of interest. Setting the initial density at the sea as the mean density ensures a proper deformation of the area of interest, and avoids arbitrary expansion of the region under the diffusion process. In this work, we propose a new method for the construction of such a sea for the diffusion.

If the simply-connected open surface $S$ is not planar, then the abovementioned curvature-based flattening methods give us an initial flattening map $\vec{r}_0 = \phi(S)$ in $\R^2$. If $S$ is initially planar, we skip the above step and set $\vec{r}_0 = S$. In other words, we treat $S$ itself as the initial flattening map. 

Now, we shrink the initial map $\vec{r}_0$ and place it inside the unit circle $\mathbb{S}^1 := \{z \in \mathbb{C}: |z|=  1\}$. Note that there will be certain gaps between the shrunk map and the circular boundary. Denote the edge length of the shrunk flattening map by $l$. We fill up the gaps using uniformly distributed points with distance $l$. This process results in an even distribution of points all over the unit disk $\mathbb{D} := \{z \in \mathbb{C}: |z| \leq 1\}$. We then triangulate the new points using the Delaunay triangulation. This gives us a triangulation $\mathbb{D}_T$ of the unit disk.

Next, we aim to construct a sea surrounding the unit disk in a natural way. Consider the reflection mapping $g:\mathbb{D} \to \mathbb{C} \setminus \mathbb{D}$ defined by
\begin{equation}
g(z) = \frac{1}{\bar{z}}.
\end{equation}
It is easy to observe that $g$ is bijective. In the discrete case, the above map sends the triangulated unit disk $\mathbb{D}_T$ to a large polygonal region $R$ in $\mathbb{C}$ with the region of $\mathbb{D}$ punctured. We now glue $\mathbb{D}_T$ and $g(\mathbb{D}_T)$ along the circular boundary $\p \mathbb{D}_T$. More explicitly, denote the glued mesh by $\widetilde{S} = (\widetilde{\mathcal{V}}, \widetilde{\mathcal{E}}, \widetilde{\mathcal{F}})$. We have
\begin{equation}
\widetilde{\mathcal{V}} = \{z\}_{z \in \mathbb{D}_T} \cup \left\{ \frac{1}{\bar{z}} \right\}_{z \in \mathbb{D}_T \setminus \p(\mathbb{D}_T)},
\end{equation}
\begin{equation}
\widetilde{\mathcal{F}} = \mathcal{F} \cup \left\{\left[\frac{1}{\bar{z}_i}, \frac{1}{\bar{z}_j}, \frac{1}{\bar{z}_k}\right]: [z_i, z_j, z_k] \in \mathcal{F} \right\},
\end{equation}
and
\begin{equation}
\widetilde{\mathcal{E}} = \{[z_i, z_j]: [z_i, z_j] \text{ is an edge of a face } T\in \widetilde{\mathcal{F}}\}.
\end{equation}
To get rid of the extremely large triangles at the outermost part of the glued mesh, we perform a simple truncation by removing the part far away from the unit disk $\mathbb{D}$. In practice, we remove all vertices and faces of $\widetilde{S}$ outside $\{z: |z| > 5\}$. Finally, we rescale the glued mesh to restore the size of the flattening map. By an abuse of notation, we continue using $\vec{r}_0$ to represent the entire region.

Now we have constructed a natural complement surrounding our region of interest in $\vec{r}_0$. We proceed to set up the density distribution at the complement part. As suggested in GN, the density at the complement part should equal the mean density at the interior part. For every face $T \in \widetilde{\mathcal{F}} \setminus \mathcal{F}$, we set
\begin{equation}
\rho^{\widetilde{\mathcal{F}}}(T) = \text{mean}_{T' \in \mathcal{F}} \rho^{\mathcal{F}}(T').
\end{equation}
This completes our construction of the sea. The above procedures are summarized in Algorithm \ref{alg:sea}. An graphical illustration of the construction is shown in Figure \ref{fig:sea}.

\begin{algorithm}[h]
\label{alg:sea}
\KwIn{An initial flattening map $\vec{r}_0$.}
\KwOut{An updated map $\vec{r}_0$ with a sea surrounding the original domain.}
\BlankLine
Shrink $\vec{r}_0$ to sit inside the unit circle $\mathbb{S}^1$\;
Fill up the gaps between the unit circle and the shrunk map by uniformly distributed points with distance $l$, where $l$ is the average edge length of $\vec{r}_0$.\;
Perform a constrained Delaunay triangulation that triangulates the unit disk with the newly added points. The connectivity of $\vec{r}_0$ is kept unchanged\;
Apply the reflection map $g(z) = \frac{1}{\overline{z}}$ to the triangulated unit disk $\mathbb{D}_T$\;
Glue $\mathbb{D}_T$ and $g(\mathbb{D}_T)$. Update $\vec{r}_0$ by the glued result\;
Remove all vertices and faces of $\vec{r}_0$ outside $\{z: |z| > 5\}$\;
Rescale $\vec{r}_0$ to restore the size of the flattening map\;
\caption{Construction of sea via reflection.}
\end{algorithm}

We now highlight the advantages of our construction of sea. One advantage of our construction is that the mesh size of the constructed sea is adaptive. Unlike the approach in GN, which used an uniform finite difference grid for the sea, our construction produces a natural distribution of points at the sea that avoids redundant computation.

More specifically, let $z_1, z_2$ be two points at the interior of the unit disk $\mathbb{D}$. It can be observed that under the reflection $z \mapsto \frac{1}{\overline{z}}$, we have
\begin{equation}
\left|\frac{1}{\overline{z_1}} - \frac{1}{\overline{z_2}}\right| = \frac{|\overline{z_1}-\overline{z_2}|}{|\overline{z_1}\overline{z_2}|} = \frac{|z_1-z_2|}{|z_1 z_2|}.
\end{equation}
This implies that if $z_1, z_2$ are located near the origin, the distance between the reflected points $\frac{1}{\overline{z_1}}$ and $\frac{1}{\overline{z_2}}$ will satisfy
\begin{equation}
\left|\frac{1}{\overline{z_1}} - \frac{1}{\overline{z_2}}\right| \gg |z_1-z_2|,
\end{equation}
since $|z_1 z_2| \ll 1$. On the other hand, if $z_1, z_2$ are located near the unit circle $\mathbb{S}^1$, we have
\begin{equation}
\left|\frac{1}{\overline{z_1}} - \frac{1}{\overline{z_2}}\right| \approx |z_1-z_2|,
\end{equation}
since $|z_1 z_2| \approx 1$.

One important consequence of the above observation is that the outermost region of the sea, which stays far away from the region of interest, consists of the coarsest triangulations. By contrast, the innermost region of the sea closest to the unit circle has the densest triangulations. This natural transition of mesh sparsity  of the sea helps reducing the number of points needed for the subsequent computation without affecting the accuracy of the result.

Another advantage of our construction is the improvement on the shape of the sea. In GN, a rectangular sea is used for the finite difference framework. The four corner regions are usually unimportant for the subsequent deformation and hence a large amount of spaces and computational efforts are wasted. By contrast, our reflection-based framework can easily overcome the above drawback. In our construction of the sea, the reflection together with the truncation produces a sea with a more regular shape. This utilizes the use of every point at the sea and prevents any redundant computations.

Finally, note that the sea is constructed simply for the diffusion process and we are only interested in the interior region. In the following discussions, by an abuse of notation, we continue using the alphabets $\mathcal{V}$, $\mathcal{F}$ without tilde whenever referring to the discrete mesh structure.

\subsection{Iterative scheme for producing density-equalizing maps} \label{sect:step3}
\begin{figure}[t!]
\centering
\includegraphics[width=0.3\textwidth]{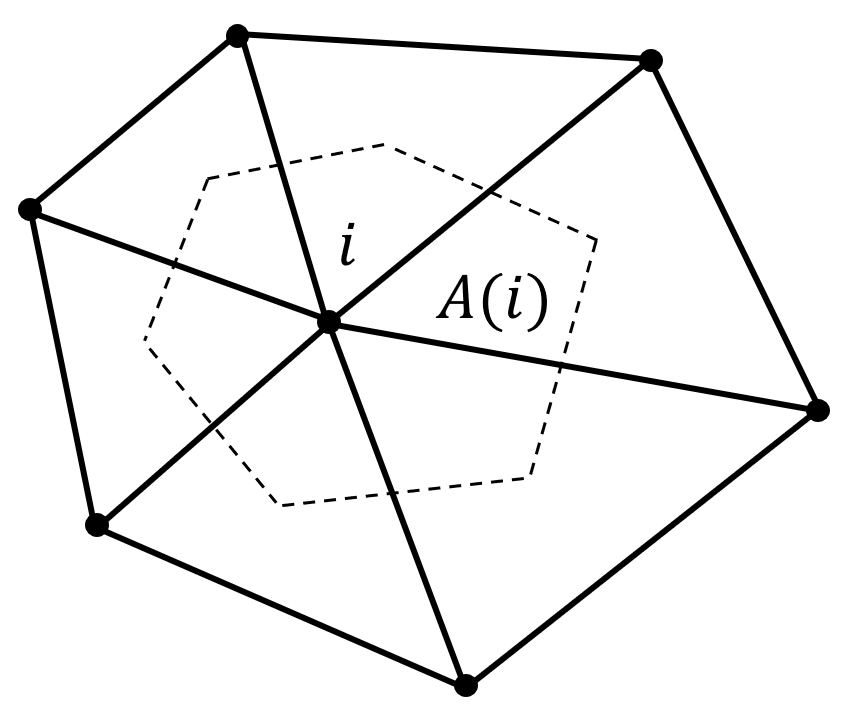}
\includegraphics[width=0.3\textwidth]{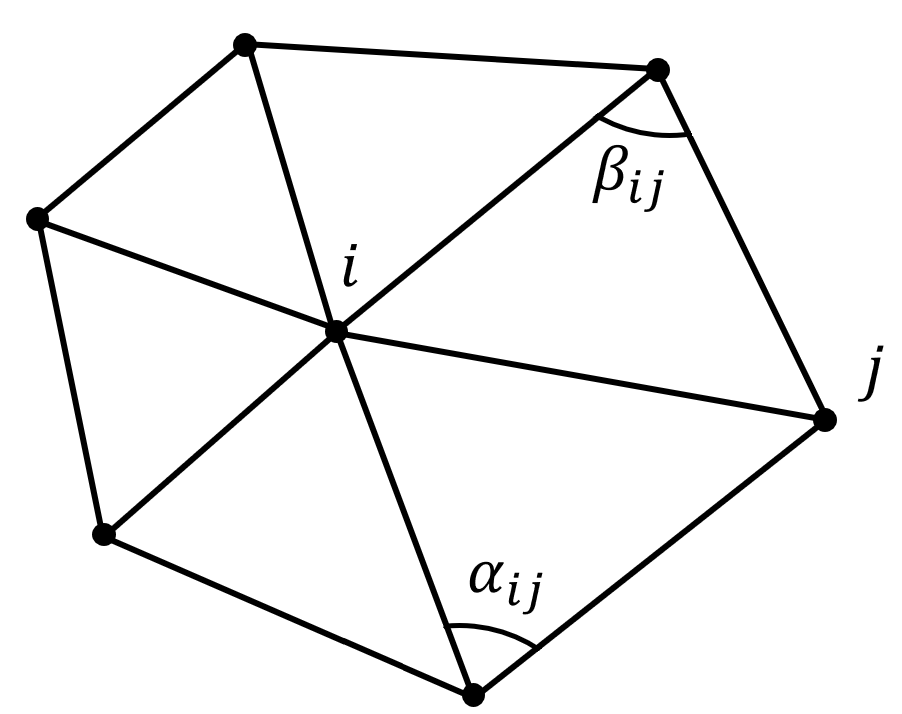}
\caption{Left: the vertex area $A(i)$ of a vertex $i$. Right: the two angles $\alpha_{ij}$ and $\beta_{ij}$ opposite to the edge $[i,j]$.}
\label{fig:laplacian}
\end{figure}
Given any simply-connected open triangular mesh, the curvature-based flattening method produces a flattened map in $\R^2$. Suppose we are given a population on each triangle element of the mesh. Define the density $\rho$ on each triangle element of the flattened map by $\frac{\text{Given population}}{\text{Area of the triangle}}$. As introduced before, after constructing the adaptive sea surrounding the map, we extend the density $\rho$ to the whole domain by setting $\rho$ at the sea to be the mean density at the original flattened map. In this subsection, we develop an iterative scheme for deforming the flattened map based on density diffusion.

To solve the diffusion equation on triangular meshes, one important issue is to discretize the Laplacian.
Let $u: \mathcal{V} \to \R$ be a function. To compute the Laplacian of $u$ at every vertex $i$, we make use of the cotangent Laplacian formulation~\cite{Pinkall93}
\begin{equation}
\Delta u(i) = \frac{1}{2A(i)} \sum_{j \in \mathcal{N}(i)} \left(\cot \alpha_{ij} + \cot \beta_{ij} \right) (u(j) - u(i)),
\end{equation}
where $A(i)$ is the vertex area of the vertex $i$, and $\alpha_{ij}$ and $\beta_{ij}$ are the two angles opposite to the edge $[i,j]$. More specifically, 
\begin{equation}
A(i) = \frac{1}{3} \sum_{T \in \{T \in \mathcal{F}: i \text{ is a vertex of the triangle }T \}} \text{Area}(T). 
\end{equation}
It is easy to observe that
\begin{equation}
\sum_{i\in \mathcal{V}} A(i) = \sum_{T \in \mathcal{F}} \text{Area}(T).
\end{equation}
This shows that the vertex area is a good discretization of the total surface area at the vertex set. The graphical illustrations of $A(i)$ and $\alpha_{ij}, \beta_{ij}$ are given in Figure \ref{fig:laplacian}.

Note that the density $\rho$ is originally defined on the triangular faces while the above $|\mathcal{V}| \times |\mathcal{V}|$ Laplacian is only applicable for vertices. To handle this discrepancy, we develop a natural transition between the value of $\rho$ on triangular faces and that on vertices. Let $\rho^{\mathcal{F}}$ and $\rho^{\mathcal{V}}$ be respectively the value of $\rho$ on faces and that on vertices in the form of column vectors. Given $\rho^{\mathcal{V}}$ on the vertices, the discretization $\rho^{\mathcal{F}}$ on the triangular faces can be obtained by considering
\begin{equation}
\rho^{\mathcal{F}} = M^{\mathcal{V}\mathcal{F}} \rho^{\mathcal{V}},
\end{equation}
where $M^{\mathcal{V}\mathcal{F}}$ is a $|\mathcal{F}| \times |\mathcal{V}|$ transition matrix defined by
\begin{equation}
M^{\mathcal{V}\mathcal{F}}_{ij} := \left\{
\begin{array}{ll}
\frac{1}{3} & \text{ if the $i$-th triangle contains the $j$-th vertex,}\\
0 & \text{ otherwise}.
\end{array}\right.
\end{equation}
Similarly, given $\rho^{\mathcal{F}}$ on the triangular faces, the discretization $\rho^{\mathcal{V}}$ on the vertices can be obtained by
\begin{equation}
\rho^{\mathcal{V}} = M^{\mathcal{F}\mathcal{V}} \rho^{\mathcal{F}},
\end{equation}
where $M^{\mathcal{F}\mathcal{V}}$ is a $|\mathcal{V}| \times |\mathcal{F}|$ transition matrix. This time, we denote
\begin{equation}
\widetilde{M}^{\mathcal{F}\mathcal{V}}_{ij} := \left\{
\begin{array}{ll}
1 & \text{ if the $j$-th triangle contains the $i$-th vertex,}\\
0 & \text{ otherwise,}
\end{array}\right.
\end{equation}
and define
\begin{equation}
M^{\mathcal{F}\mathcal{V}} := \begin{pmatrix}
\widetilde{M}^{\mathcal{F}\mathcal{V}}_{1,:}/\|\widetilde{M}^{\mathcal{F}\mathcal{V}}_{1,:}\|_0 \\
\widetilde{M}^{\mathcal{F}\mathcal{V}}_{2,:}/\|\widetilde{M}^{\mathcal{F}\mathcal{V}}_{2,:}\|_0 \\
\vdots\\
\widetilde{M}^{\mathcal{F}\mathcal{V}}_{|\mathcal{F}|,:}/\|\widetilde{M}^{\mathcal{F}\mathcal{V}}_{|\mathcal{F}|,:}\|_0 \\
\end{pmatrix}.
\end{equation}
It is easy to check that
\begin{equation}
{M}^{\mathcal{F}\mathcal{V}}{M}^{\mathcal{V}\mathcal{F}} = |\mathcal{F}|,
\end{equation}
and
\begin{equation}
{M}^{\mathcal{V}\mathcal{F}}{M}^{\mathcal{F}\mathcal{V}} = |\mathcal{V}|.
\end{equation}
Therefore, the two operators ${M}^{\mathcal{F}\mathcal{V}}$ and ${M}^{\mathcal{V}\mathcal{F}}$ are consistent with each other.

After discretizing the Laplacian, we propose the following semi-discrete backward Euler method for solving the diffusion equation \eqref{eqt:diffusion_continuous}:
\begin{equation}
\frac{\rho^{\mathcal{V}}_{n} - \rho^{\mathcal{V}}_{n-1}}{\delta t} = \Delta_{n-1} \rho^{\mathcal{V}}_{n}.  
\end{equation}
Here $\rho^{\mathcal{V}}_n$ is the value of $\rho$ on the vertices at the $n$-th iteration, $\Delta_n$ is the cotangent Laplacian of the deformed map $\vec{r}_{n}$, and $\delta t$ is the time step for the iterations. By a simple rearrangement, the above equation is equivalent to
\begin{equation} \label{eqt:update_rho}
\rho^{\mathcal{V}}_{n} = (I - \delta t \Delta_{n-1})^{-1} \rho^{\mathcal{V}}_{n-1}.
\end{equation}
Note that the above semi-discrete backward Euler method is unconditionally stable and ensures the convergence of our algorithm. Also, the cotangent Laplacian $\Delta_n$ is a symmetric positive definite matrix and hence \eqref{eqt:update_rho} can be efficiently solved by numerical solvers.

After discretizing the diffusion equation, we consider the production of the induced vector field. We first need to discrete the gradient operator $\nabla$. Consider the face-based discretization $(\nabla \rho)_n^{\mathcal{F}}(T)$ defined on every triangle element $T = [i,j,k]$ at the $n$-th iteration. Note that $(\nabla \rho)_n^{\mathcal{F}}(T)$ should satisfy
\begin{equation} \label{eqt:gradient_property}
\left\{\begin{array}{l}
\langle (\nabla \rho)_n^{\mathcal{F}}(T), e_{jk} \rangle = \rho_n^{\mathcal{V}}(k) - \rho_n^{\mathcal{V}}(j),\\
\langle (\nabla \rho)_n^{\mathcal{F}}(T), e_{ki} \rangle = \rho_n^{\mathcal{V}}(i) - \rho_n^{\mathcal{V}}(k),\\
\langle (\nabla \rho)_n^{\mathcal{F}}(T), e_{ij} \rangle = \rho_n^{\mathcal{V}}(j) - \rho_n^{\mathcal{V}}(i),\\
\langle (\nabla \rho)_n^{\mathcal{F}}(T),N \rangle = 0,\\
\end{array}\right.
\end{equation}
where $e_{ij} = [i,j], e_{jk} = [j,k], e_{ki} = [k,i]$ are the three directed edges of $T$ in the form of vectors, and $N$ is a unit normal vector of $T$. Since
\begin{equation}
e_{jk} + e_{ki} + e_{ij} = 0,
\end{equation}
the third equation in \eqref{eqt:gradient_property} automatically follows from the first two equations. Note that \eqref{eqt:gradient_property} can be solved analytically with the solution
\begin{equation} \label{eqt:discrete_gradient}
(\nabla\rho)_n^{\mathcal{F}}(T) = -\frac{1}{2\text{Area}(T)} N \times \left(\rho_n^{\mathcal{V}}(i) e_{jk} + \rho_n^{\mathcal{V}}(j) e_{ki} + \rho_n^{\mathcal{V}}(k) e_{ij}\right).
\end{equation}
This gives us an accurate approximation of the gradient operator on triangulated surfaces.

Note that the above approximation is developed on the triangle elements but not on the vertices. For the face-to-vertex conversion, we again make use of a matrix multiplication. Note that the previously developed matrices ${M}^{\mathcal{F}\mathcal{V}}$ and ${M}^{\mathcal{V}\mathcal{F}}$ are purely combinatorial as the directions are not important in their uses. By contrast, the directions are important for computing the gradient operator. Therefore, we need to take the geometry of the mesh into account in designing the conversion matrix for the gradient operator. To emphasize the weight of different directions in the conversion, we use the triangle area as a weight function. More specifically, we denote 
\begin{equation}
\widetilde{W}^{\mathcal{F}\mathcal{V}}_{ij} = \left\{
\begin{array}{ll}
\text{Area}(T_j) & \text{ if the $j$-th triangle $T_j$ contains the $i$-th vertex,}\\
0 & \text{ otherwise.}
\end{array}\right.
\end{equation}
Then we can define
\begin{equation}
W^{\mathcal{F}\mathcal{V}} := \begin{pmatrix}
\widetilde{W}^{\mathcal{F}\mathcal{V}}_{1,:}/\|\widetilde{W}^{\mathcal{F}\mathcal{V}}_{1,:}\|_1 \\
\widetilde{W}^{\mathcal{F}\mathcal{V}}_{2,:}/\|\widetilde{W}^{\mathcal{W}\mathcal{V}}_{2,:}\|_1 \\
\vdots\\

\widetilde{W}^{\mathcal{F}\mathcal{V}}_{|\mathcal{F}|,:}/\|\widetilde{W}^{\mathcal{F}\mathcal{V}}_{|\mathcal{F}|,:}\|_1 \\
\end{pmatrix}.
\end{equation}
With this weighted face-to-vertex conversion matrix, we have
\begin{equation}
(\nabla\rho)^{\mathcal{V}}_{n} := W^{\mathcal{F}\mathcal{V}}(\nabla\rho)^{\mathcal{F}}_{n}.
\end{equation}
With all differential operators discretized, we are now ready to introduce our iterative scheme for computing density-equalizing maps. In each iteration, we update the density by solving \eqref{eqt:update_rho} and compute the induced gradient $(\nabla\rho)_n^{\mathcal{V}}$ based on the abovementioned procedures. Then, we deform the map by
\begin{equation}
 \vec{r}_{n} = \vec{r}_{n-1} + \delta t (\nabla\rho)_{n}^{\mathcal{V}}.
\end{equation}

For the stopping criterion, we consider the quantity $\frac{\text{sd}(\rho^{\mathcal{F}}_n)}{\text{mean}(\rho^{\mathcal{F}}_n)}$. Note that the standard deviation $\text{sd}(\rho^{\mathcal{F}}_n)$ measures the dispersion of the updated density $\rho^{\mathcal{F}}_n$, and we normalize it using ${\text{mean}(\rho^{\mathcal{F}}_n)}$ to remove the effect of arbitrary scaling of $\rho^{\mathcal{F}}_n$. Also, it is easy to note that $\frac{\text{sd}(\rho^{\mathcal{F}}_n)}{\text{mean}(\rho^{\mathcal{F}}_n)} = 0$ if and only if the density is completely equalized. Hence, this normalized quantity can be used for determining the convergence of the iterative algorithm. Finally, we rescale the mapping result so that the total area of $S$ is preserved under our density-equalizing mapping algorithm. 

We remark that the step size $\delta t$ affects the convergence rate of the algorithm. By dimensional analysis on the diffusion equation \eqref{eqt:diffusion_continuous}, an appropriate dimension of $\delta t$ would be $L^2$. Also, note that $\delta t$ should be independent to the magnitude of $\rho$. Therefore, a reasonable choice of $\delta t$ is
\begin{equation}
\delta t = \text{min}\left\{\frac{\text{min}(\rho^{\mathcal{F}}_0)}{\text{mean}(\rho^{\mathcal{F}}_0)}, \frac{\text{mean}(\rho^{\mathcal{F}}_0)}{\text{max}(\rho^{\mathcal{F}}_0)}\right\} \times \text{Area}(S).
\end{equation}
The first term is a dimensionless quantity that takes extreme relative density ratios into account, and the second term is a natural quantity with dimension $L^2$. Algorithm \ref{alg:main} summarizes our proposed method for producing density-equalizing maps of simply-connected open surfaces.

\begin{algorithm}[h]
\label{alg:main}
\KwIn{A simply-connected open triangulated surface $S$, a population on each triangle, and a stopping parameter $\epsilon$.}
\KwOut{A density-equalizing flattening map $f:S \to \R^2$.}
\BlankLine
\uIf{$S$ is planar}{
    Set $\vec{r}_0 = S$ \;
  }
\Else{
    Compute a curvature-based flattening map $\phi: S \to \mathbb{C}$ using Algorithm \ref{alg:tutte_curvature} or Algorithm \ref{alg:chi_curvature}. Denote $\vec{r}_0 = \phi(S)$\;
}
Define the density $\rho^{\mathcal{F}}_0 = \frac{\text{Given population}}{\text{Area of the triangle}}$ on each triangle of $\vec{r}_0$\;
Update $\vec{r}_0$ with an adaptive sea constructed using Algorithm \ref{alg:sea}\;
Extend $\rho^{\mathcal{F}}_0$ to the whole domain by setting $\rho^{\mathcal{F}}_0$ at the sea to be the mean of the original $\rho^{\mathcal{F}}_0$\;
Compute $\rho^{\mathcal{V}}_0 = {M}^{\mathcal{F}\mathcal{V}} \rho^{\mathcal{F}}_0$\;
Set $\delta t = \text{min}\left\{\frac{\text{min}(\rho^{\mathcal{F}}_0)}{\text{mean}(\rho^{\mathcal{F}}_0)}, \frac{\text{mean}(\rho^{\mathcal{F}}_0)}{\text{max}(\rho^{\mathcal{F}}_0)}\right\} \times \text{Area}(S)$\;
Set $n = 0$\;
\Repeat{$\frac{\text{sd}(\rho^{\mathcal{F}}_n)}{\text{mean}(\rho^{\mathcal{F}}_n)} < \epsilon$}{ 
Update $n = n + 1$\;
Solve $\rho^{\mathcal{V}}_{n} = (I - \delta t \Delta_{n-1})^{-1} \rho^{\mathcal{V}}_{n-1}$\;
Compute the face-based discrete gradient $(\nabla\rho)_n^{\mathcal{F}}(T) = -\frac{1}{2\text{Area}(T)} N \times \left(\rho_n^{\mathcal{V}}(i) e_{jk} + \rho_n^{\mathcal{V}}(j) e_{ki} + \rho_n^{\mathcal{V}}(k) e_{ij}\right)$\;
Perform the conversion $(\nabla\rho)_n^{\mathcal{V}} = W^{\mathcal{F}\mathcal{V}}(\nabla\rho)_n^{\mathcal{F}}$\;
Update $\vec{r}_{n} = \vec{r}_{n-1} + \delta t (\nabla\rho)_n^{\mathcal{V}}$\;
Compute $\rho^{\mathcal{F}}_{n} = {M}^{\mathcal{V}\mathcal{F}} \rho^{\mathcal{V}}_{n}$\;
}
Obtain $f(S) = \vec{r}_{n} \times \frac{\text{Area}(\vec{r}_{0})}{\text{Area}(\vec{r}_{n})}$\;
\caption{Density-equalizing map for simply-connected open surfaces}
\end{algorithm}

\subsection{The choice of population and its effects}
Before ending this section, we discuss the choice of the initial population and its effect on the final result obtained by our algorithm. Some choices and the corresponding effects are listed below:
\begin{enumerate}[(i)]
 \item If we set a relatively high population at a certain region of the input surface, the population will cause an expansion during the density-equalization. The region will be magnified in the final density-equalizing mapping result. 
 \item Similarly, if we set a relatively low population at a certain region of the input surface, the region will shrink in the final density-equalizing mapping result.
 \item If we set the population to be the area of every triangle element of the input surface, the resulting density-equalizing map will be an area-preserving planar parameterization of the input surface as we have
 \begin{equation}
  \frac{\text{Initial area}}{\text{Final area}} = \frac{\text{Given population}}{\text{Final area}} = \text{Density} = \text{Constant}.
 \end{equation}
\end{enumerate}
Examples are given in Section \ref{sect:experiment} to illustrate the effect of different input populations.

\section{Experimental results}\label{sect:experiment}
In this section, we demonstrate the effectiveness of our proposed algorithm using various experiments. Our algorithms are implemented in MATLAB. The linear systems in our algorithm are solved using the backslash operator in MATLAB. All experiments are performed on a PC with Intel i7-6700K CPU and 16~GB RAM. All surfaces are discretized in the form of triangular meshes. In all experiments, the stopping parameter $\epsilon$ is set to be $10^{-3}$. Some of the surface meshes are adapted from the AIM@SHAPE Shape Repository~\cite{aimatshape}. 

\begin{figure}[t!]
\centering
\includegraphics[width=0.48\textwidth]{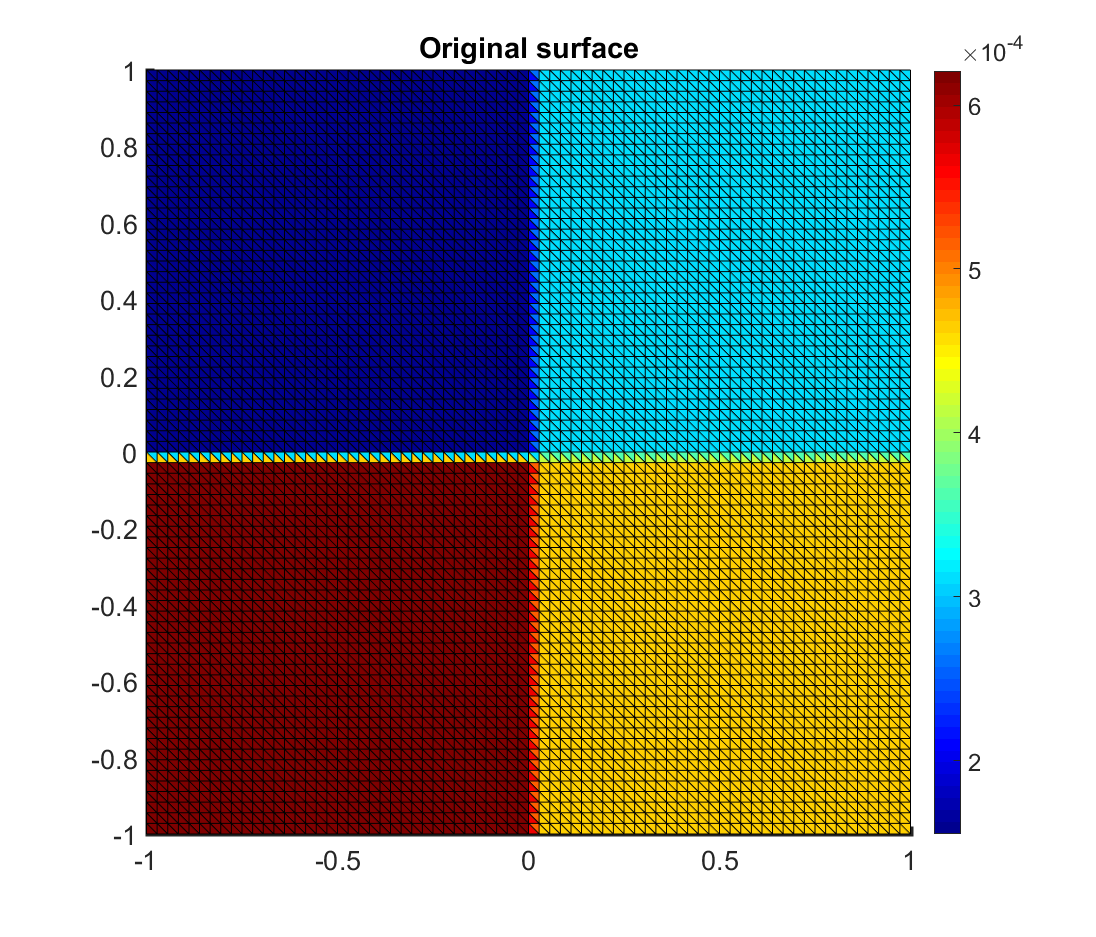}
\includegraphics[width=0.48\textwidth]{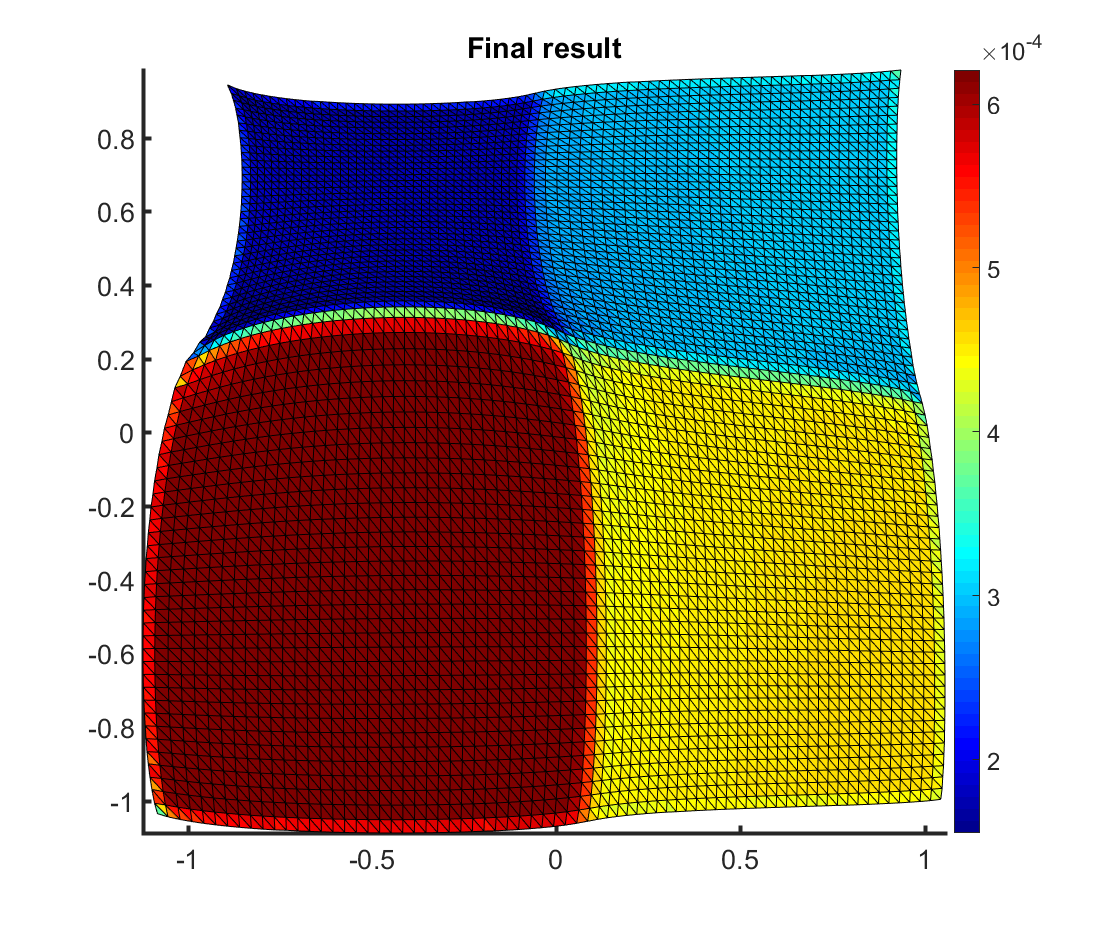}
\includegraphics[width=0.32\textwidth]{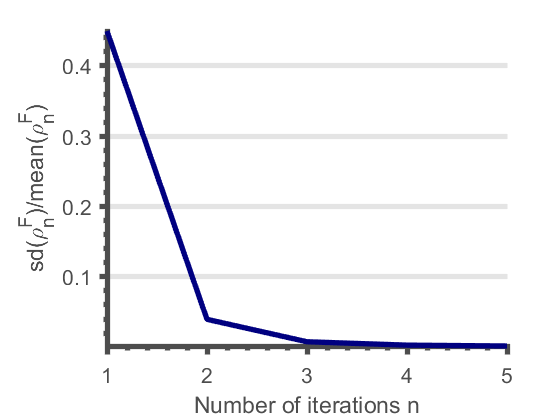}
\includegraphics[width=0.32\textwidth]{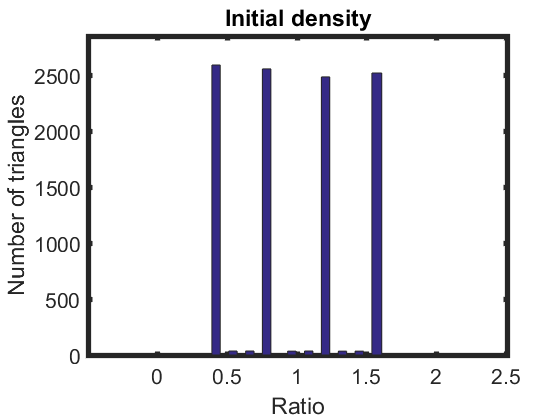}
\includegraphics[width=0.32\textwidth]{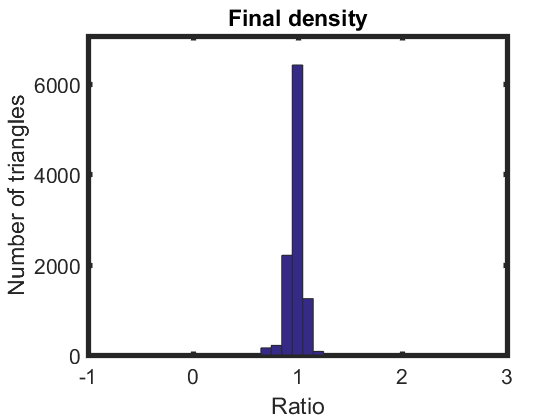}
\caption{Density-equalization on a square. Top left: the initial shape colored with a given population distribution. Top right: the density-equalizing map colored with the final area of each triangle element. Bottom left: the values of $\frac{\text{sd}(\rho^{\mathcal{F}}_n)}{\text{mean}(\rho^{\mathcal{F}}_n)}$. Bottom middle: the histogram of the initial density $\frac{\text{Given population}}{\text{Initial area}}$ on each triangle element. Bottom right: the histogram of the final density $\frac{\text{Given population}}{\text{Final area}}$ on each triangle element. }
\label{fig:square}
\end{figure}

\begin{figure}[t!]
\centering
\includegraphics[width=0.48\textwidth]{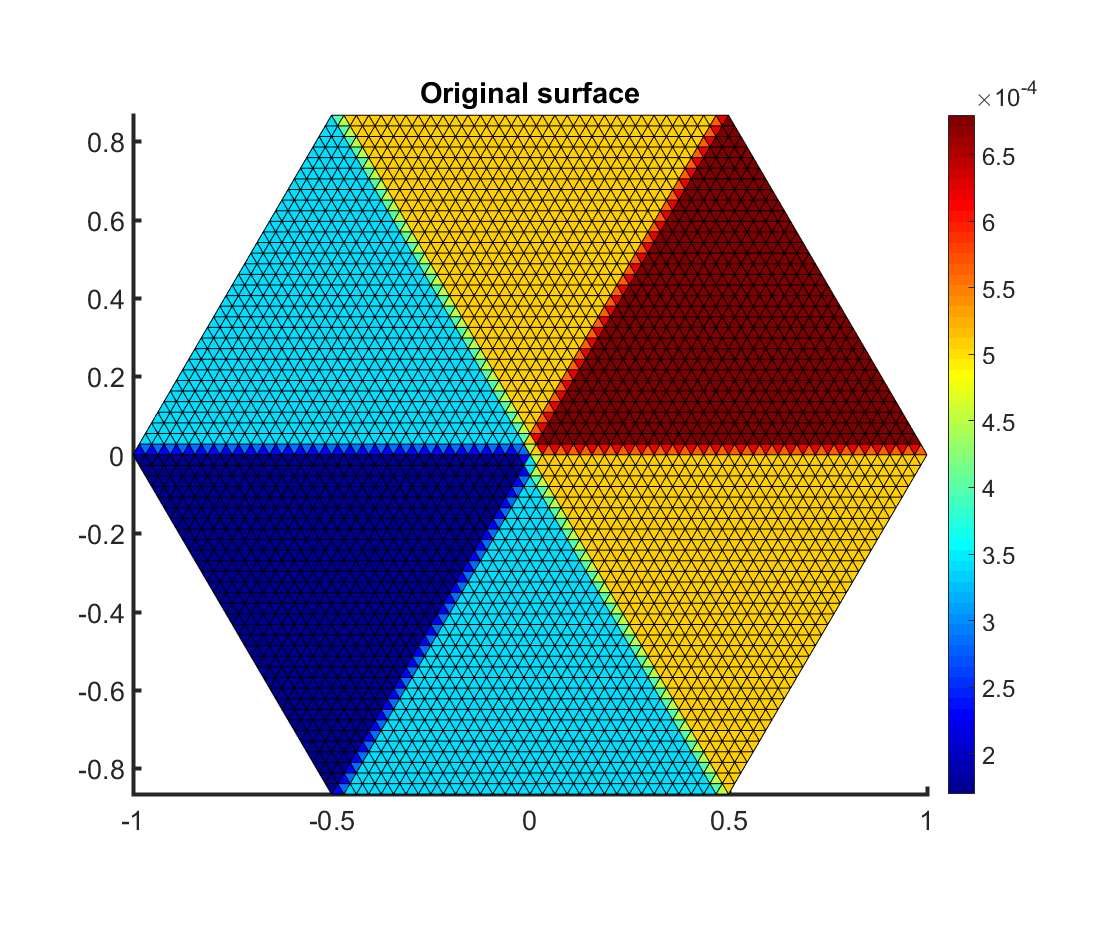}
\includegraphics[width=0.48\textwidth]{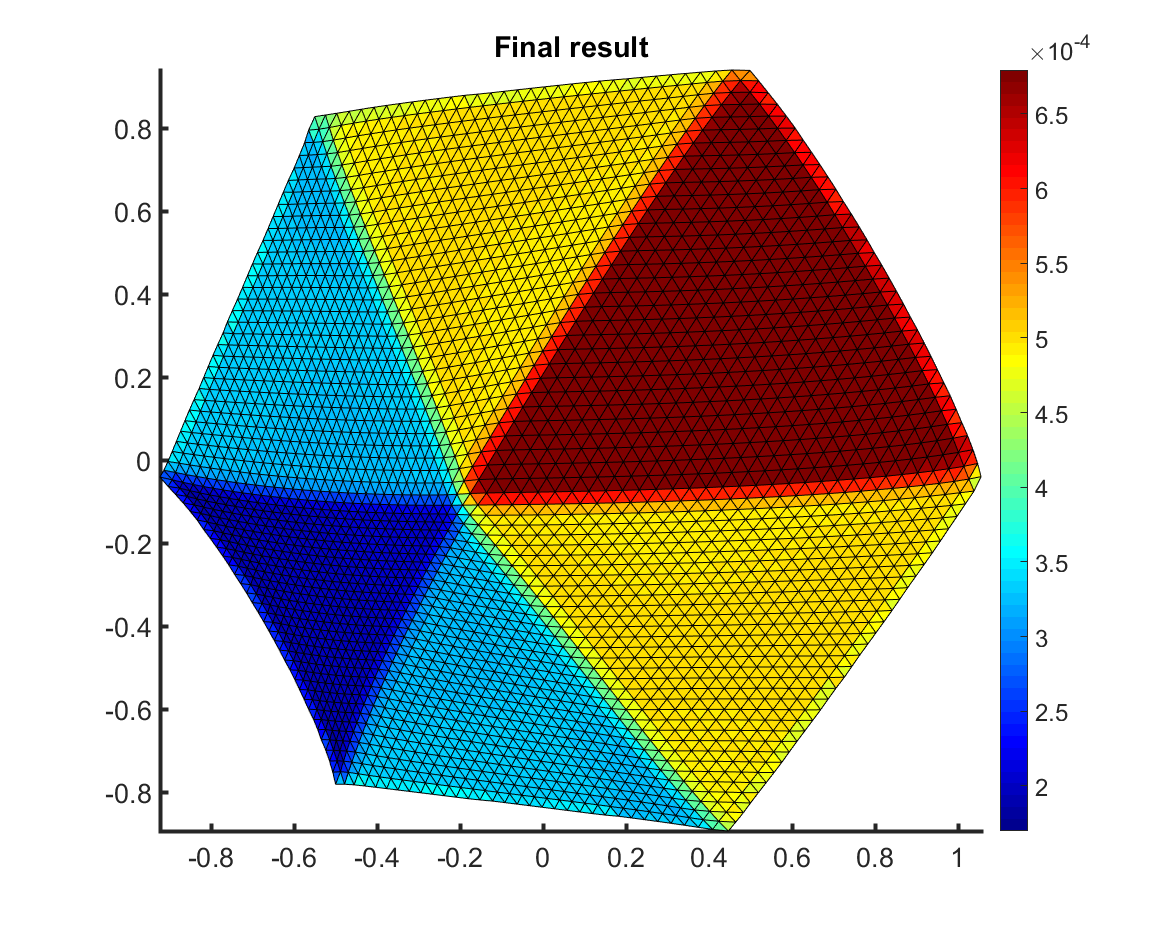}
\includegraphics[width=0.32\textwidth]{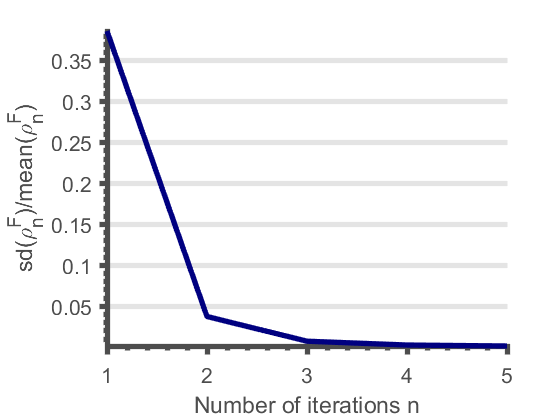}
\includegraphics[width=0.32\textwidth]{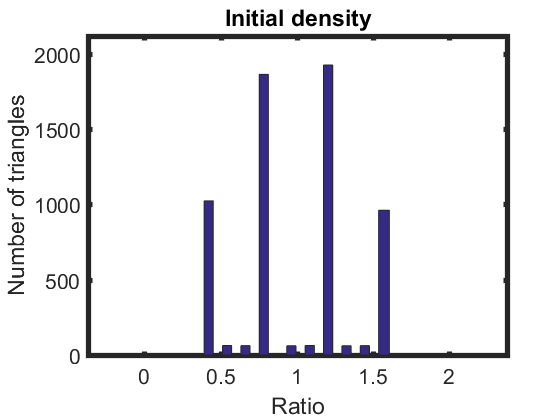}
\includegraphics[width=0.32\textwidth]{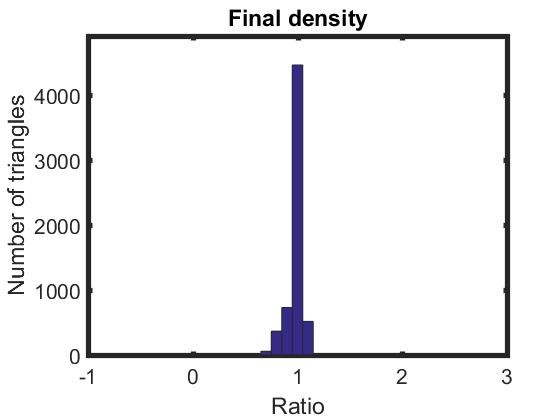}
\caption{Density-equalization on a hexagon. Top left: the initial shape colored with a given population distribution. Top right: the density-equalizing map colored with the final area of each triangle element.  Bottom left: the values of $\frac{\text{sd}(\rho^{\mathcal{F}}_n)}{\text{mean}(\rho^{\mathcal{F}}_n)}$. Bottom middle: the histogram of the initial density $\frac{\text{Given population}}{\text{Initial area}}$ on each triangle element. Bottom right: the histogram of the final density $\frac{\text{Given population}}{\text{Final area}}$ on each triangle element.}
\label{fig:hexagon}
\end{figure}

\subsection{Examples of density-equalizing maps produced by our algorithm}
We begin with two synthetic examples of regular polygons on $\R^2$. Figures \ref{fig:square} and \ref{fig:hexagon} respectively show a square and a hexagon with a given population on every triangle element, and the density-equalizing results obtained by our proposed algorithm. In both examples the final densities $\frac{\text{Given population}}{\text{Final area}}$ highly concentrate at $1$, meaning that the densities are well equalized. Also, the plots of the quantity $\frac{\text{sd}(\rho^{\mathcal{F}}_n)}{\text{mean}(\rho^{\mathcal{F}}_n)}$ show that iterative scheme converges rapidly.

\begin{figure}[t!]
\centering
\includegraphics[width=0.325\textwidth]{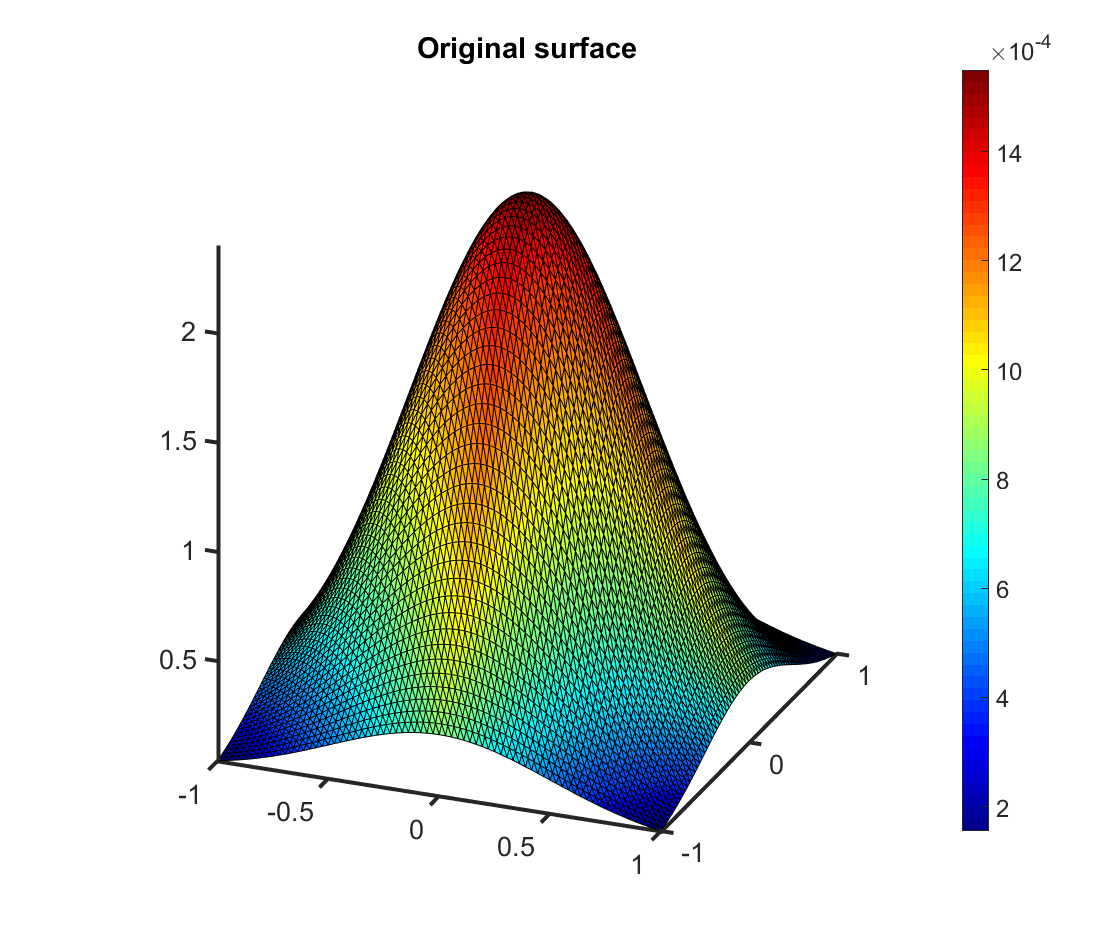}
\includegraphics[width=0.325\textwidth]{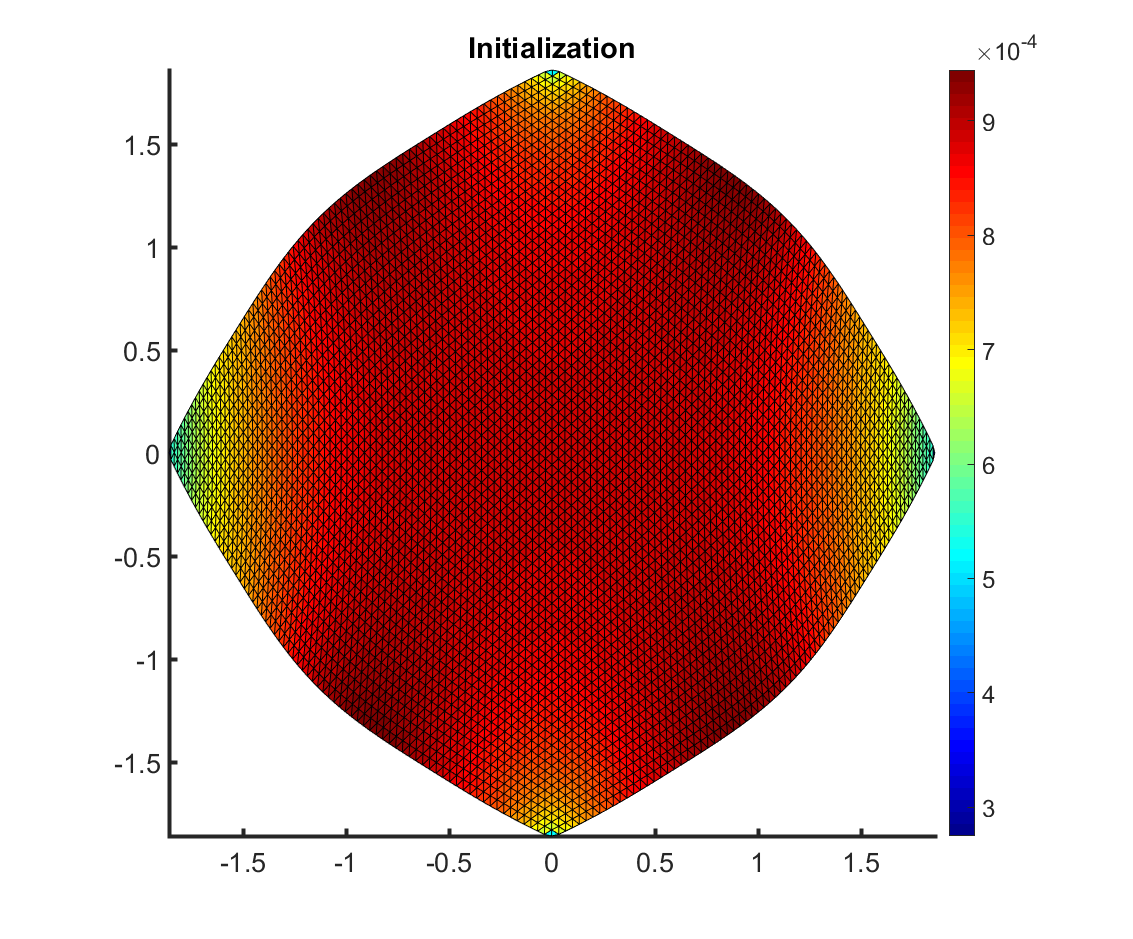}
\includegraphics[width=0.32\textwidth]{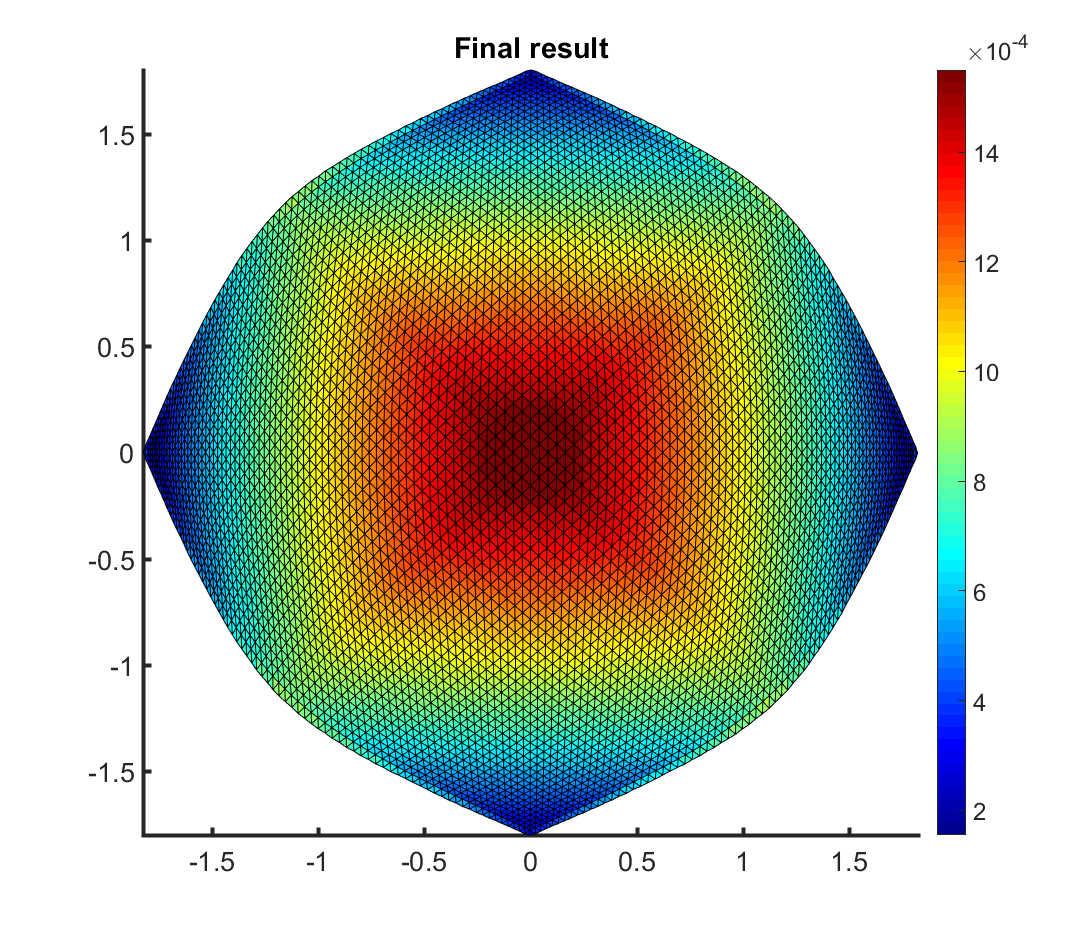}
\includegraphics[width=0.32\textwidth]{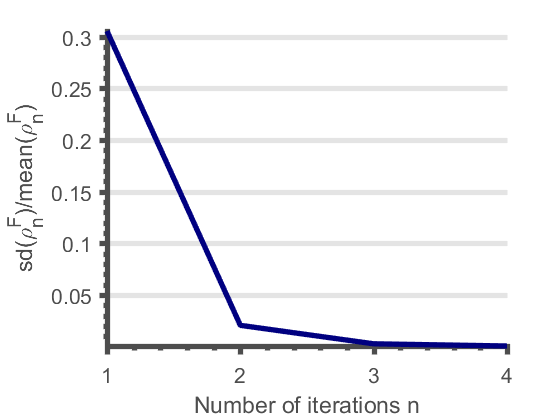}
\includegraphics[width=0.32\textwidth]{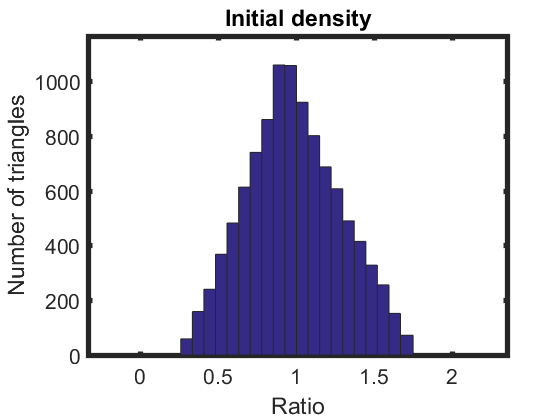}
\includegraphics[width=0.32\textwidth]{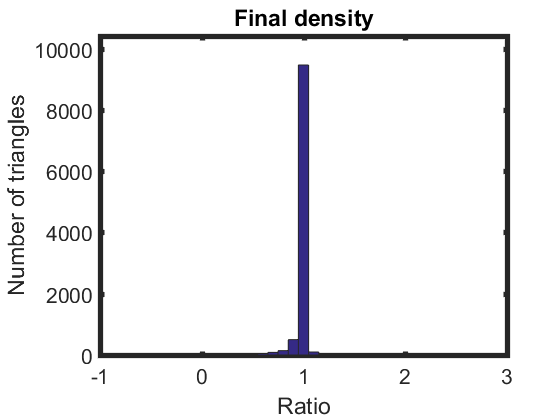}
\caption{Density-equalizing map for a surface in $\R^3$ with Gaussian shape. Top left: the initial shape colored with a given population distribution. Top middle: the curvature-based Tutte flattening initialization colored with the area of each flattened triangle element. Top right: the final density-equalizing map colored with the final area of each triangle element. Bottom right: the values of $\frac{\text{sd}(\rho^{\mathcal{F}}_n)}{\text{mean}(\rho^{\mathcal{F}}_n)}$. Bottom middle: the histogram of the density $\frac{\text{Given population}}{\text{Initial flattened area}}$ on each flattened triangle element after the Tutte flattening initialization. Bottom right: the histogram of the density $\frac{\text{Given population}}{\text{Final area}}$ on each triangle element of the final result.}
\label{fig:gaussian}
\end{figure}

We then consider a synthetic example of a surface in $\R^3$ with Gaussian shape. The domain of the shape is $[0,1] \times [0,1]$ and the population is set to be $2.2 - |x| - |y|$, where $(x,y)$ are the $x$- and $y$-coordinates of the centroid of each triangle element. Algorithm \ref{alg:tutte_curvature} is used for the initialization of the density-equalization algorithm. Figure \ref{fig:gaussian} shows the initial surface and the mapping result obtained by our density-equalizing mapping algorithm. The plots indicate that the density is well equalized by our algorithm.

\begin{figure}[t!]
\centering
\includegraphics[width=0.32\textwidth]{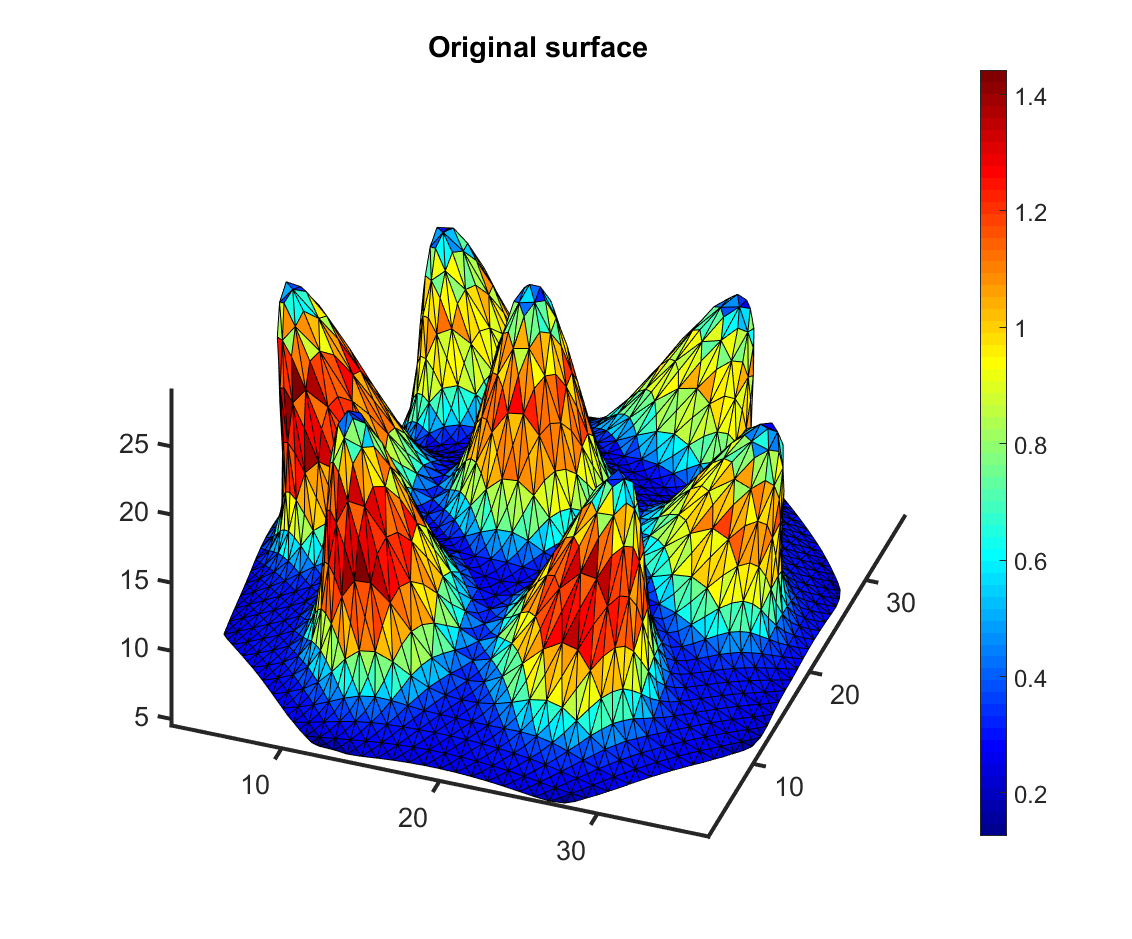}
\includegraphics[width=0.32\textwidth]{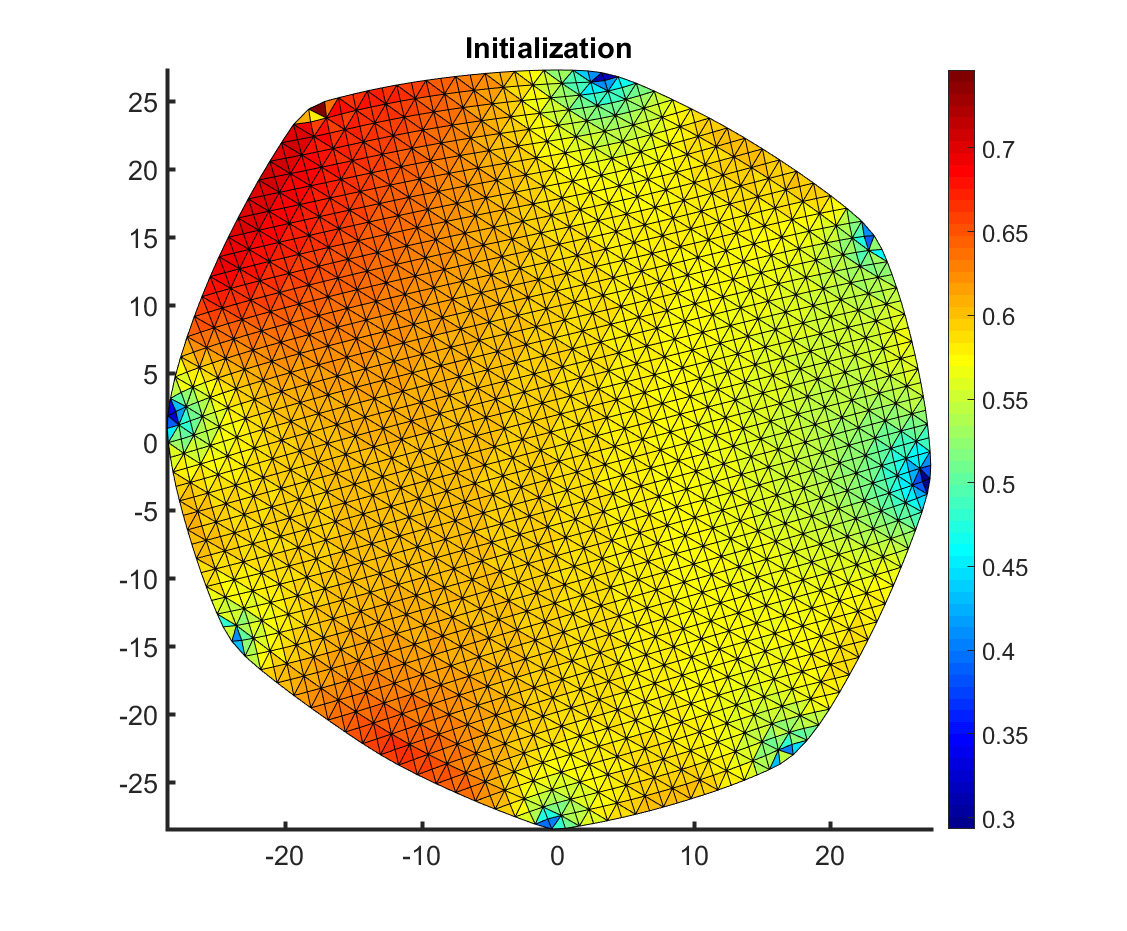}
\includegraphics[width=0.32\textwidth]{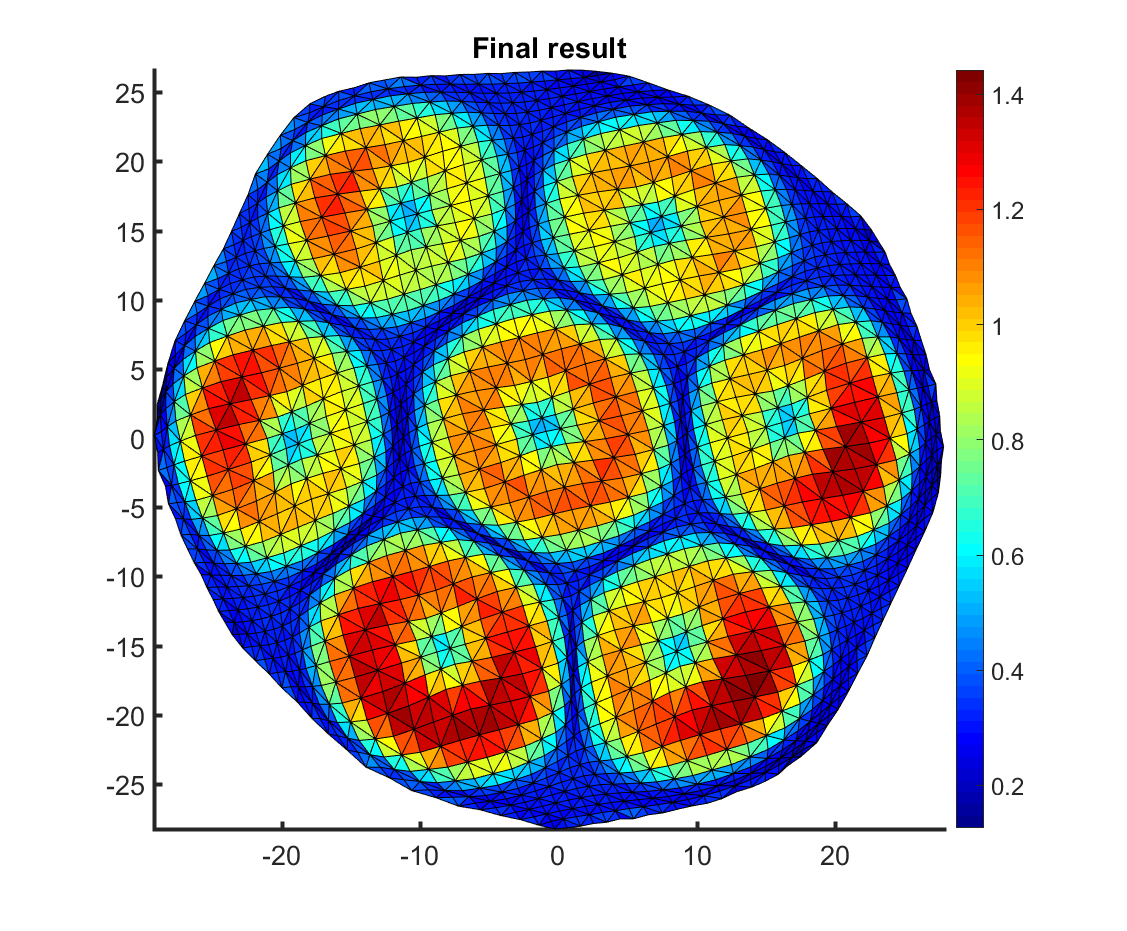}
\includegraphics[width=0.32\textwidth]{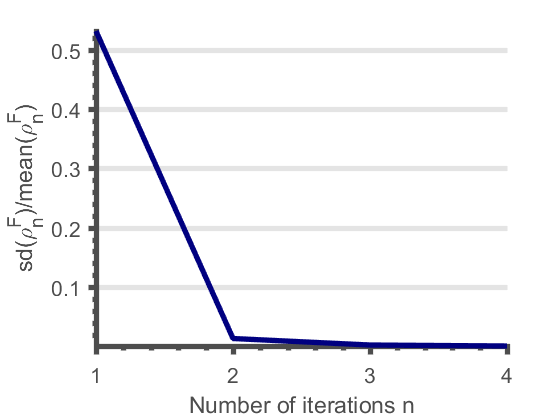}
\includegraphics[width=0.32\textwidth]{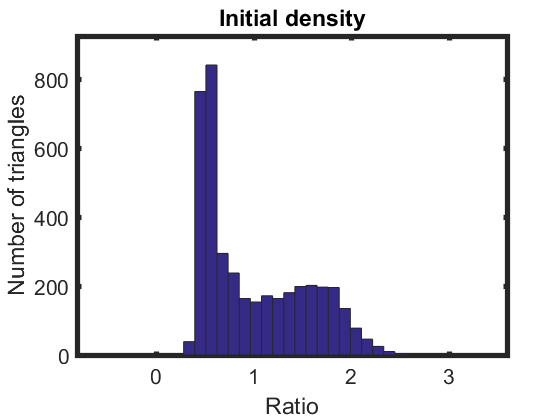}
\includegraphics[width=0.32\textwidth]{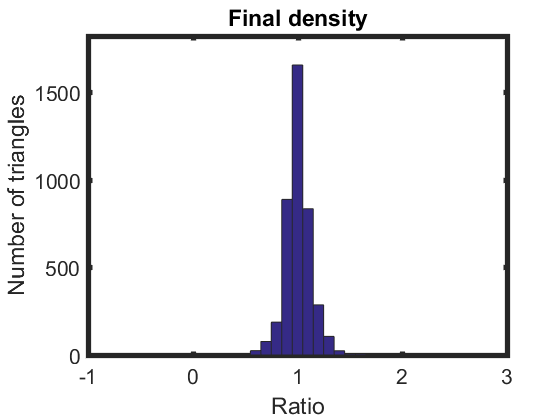}
\caption{Area-preserving parameterization of a surface with multiple peaks in $\R^3$. Top left: the initial shape colored with the initial area of each triangle element. Top middle: the curvature-based Tutte flattening initialization colored with the area of each flattened triangle element. Top right: the final density-equalizing map colored with the final area of each triangle element.  Bottom right: the values of $\frac{\text{sd}(\rho^{\mathcal{F}}_n)}{\text{mean}(\rho^{\mathcal{F}}_n)}$. Bottom middle: the histogram of the density $\frac{\text{Initial area}}{\text{Initial flattened area}}$ on each flattened triangle element after the Tutte flattening initialization. Bottom right: the histogram of the density $\frac{\text{Initial area}}{\text{Final area}}$ on each triangle element of the final result.}
\label{fig:peaks}
\end{figure}

\begin{figure}[t!]
\centering
\includegraphics[width=0.3\textwidth]{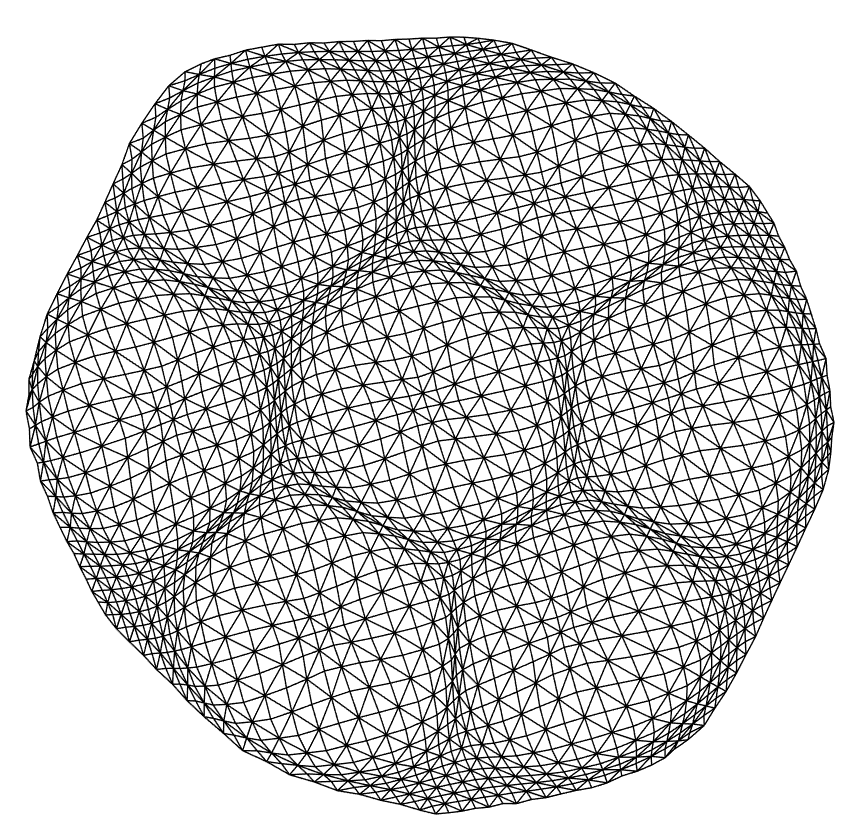}
\includegraphics[width=0.35\textwidth]{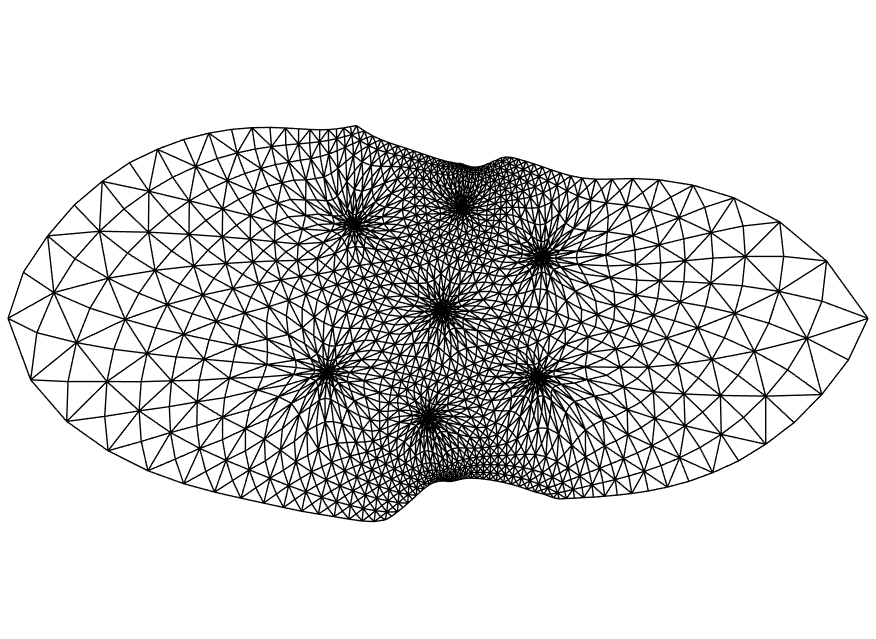}
\includegraphics[width=0.3\textwidth]{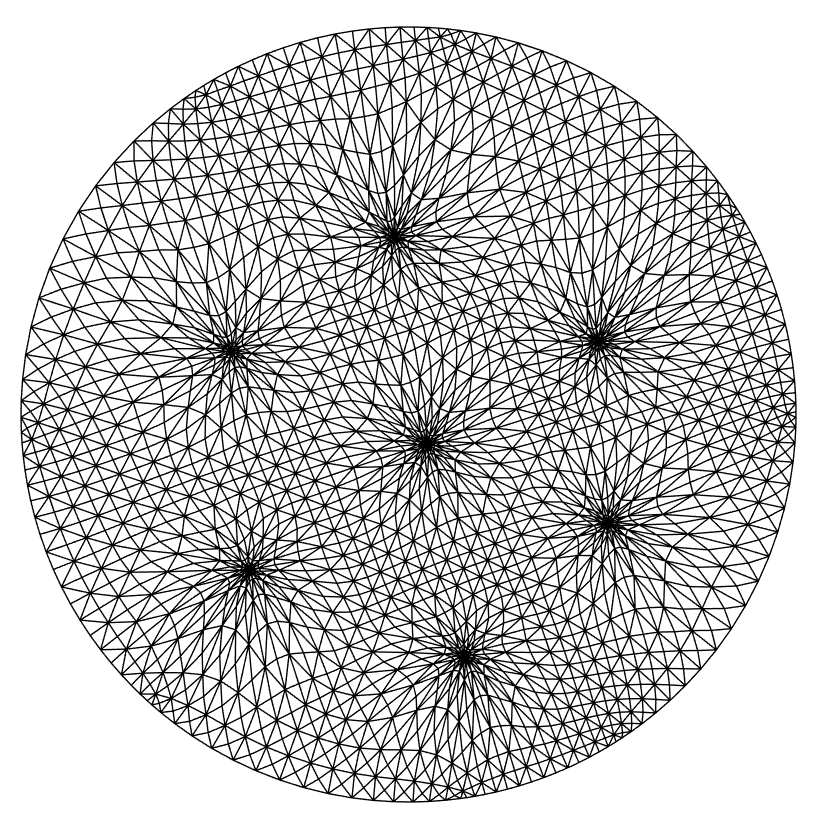}
\caption{Comparison of different parameterization schemes for a surface with multiple peaks in $\R^3$ shown in Figure \ref{fig:peaks}. Left: The area-preserving parameterization by our method. Middle: The free-boundary conformal parameterization by Desbrun et al.~\cite{Desbrun02}. Right: The disk conformal parameterization by Choi and Lui~\cite{Choi17}.}
\label{fig:peaks_comparison}
\end{figure}

We consider another synthetic example of a surface with multiple peaks in $\R^3$. This time, we set the population as the area of each triangle element on the initial surface. In other words, our proposed algorithm should result in an area-preserving flattening map. Again, Algorithm \ref{alg:tutte_curvature} is used for the initialization of the density-equalization algorithm. Figure \ref{fig:peaks} shows the initial surface and the mapping result obtained by our density-equalizing mapping algorithm. The flattening map effectively preserves the area ratios. We compare our density-equalizing mapping result with the state-of-the-art conformal parameterization algorithms~\cite{Desbrun02,Choi17}. Figure \ref{fig:peaks_comparison} shows the parameterization results, whereby the peaks are substantially shrunk for conformal parameterizations, and the boundary of the free-boundary conformal parameterization is significantly different from that of the original surface. By contrast, the peaks are flattened without being shrunk under our proposed algorithm.

\begin{figure}[t!]
\centering
\includegraphics[width=0.32\textwidth]{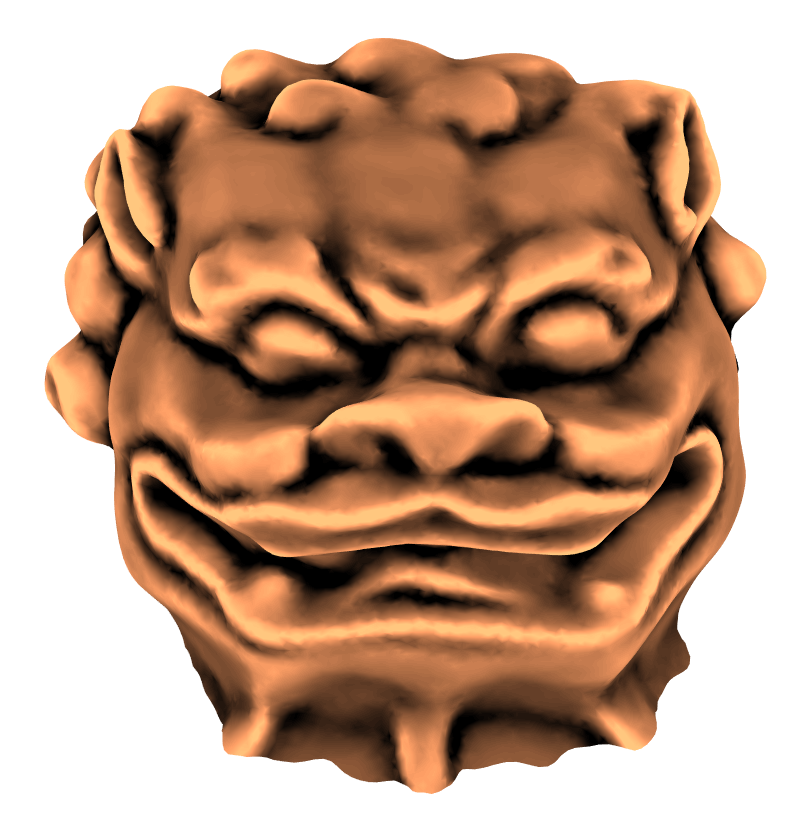}
\includegraphics[width=0.3\textwidth]{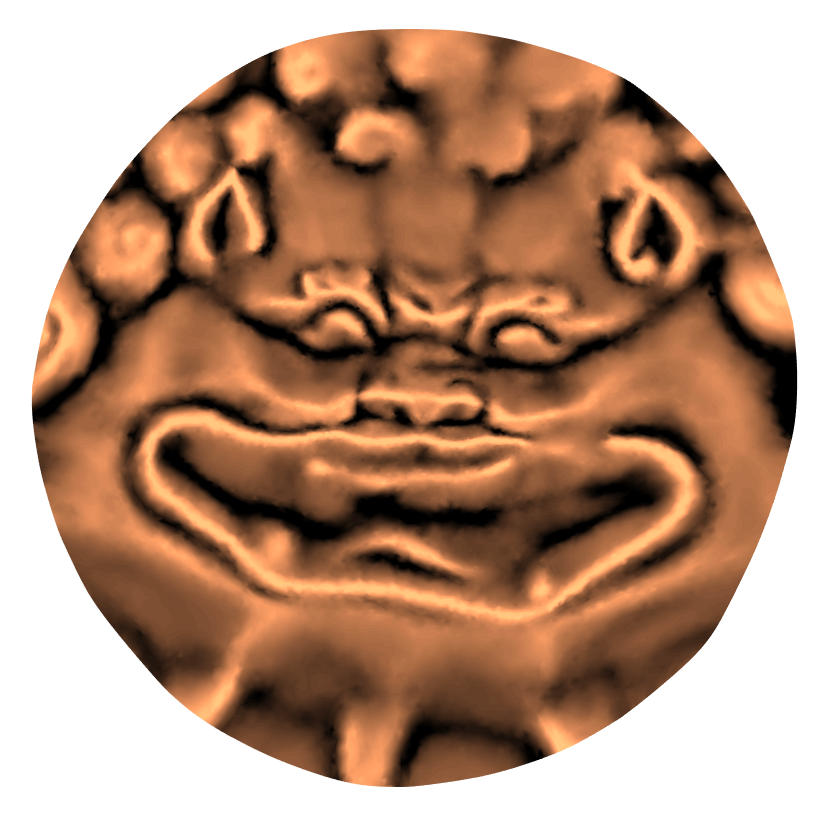}
\includegraphics[width=0.3\textwidth]{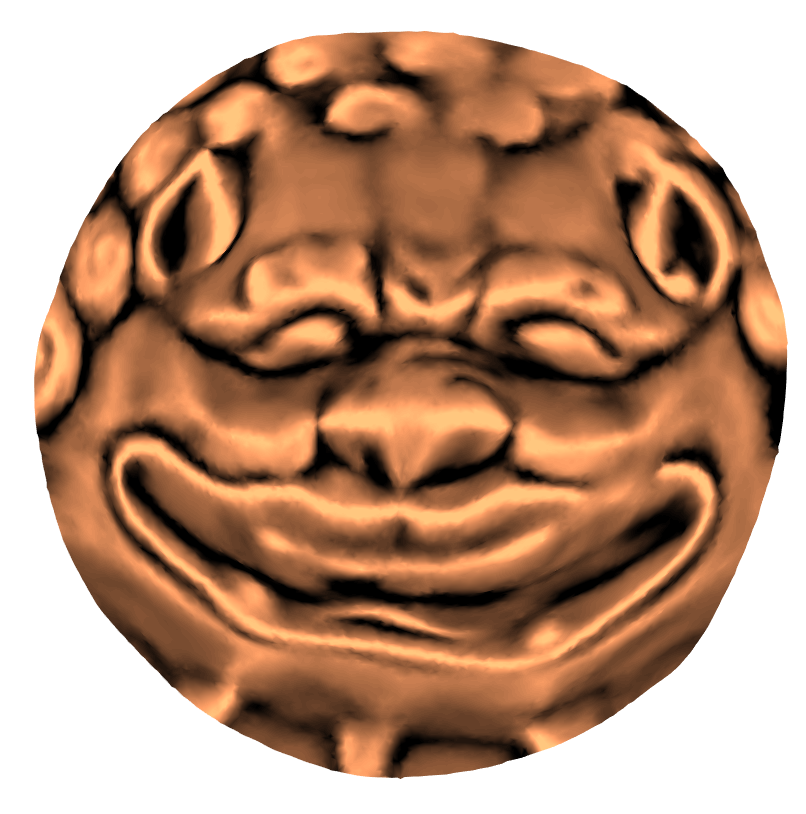}
\includegraphics[width=0.32\textwidth]{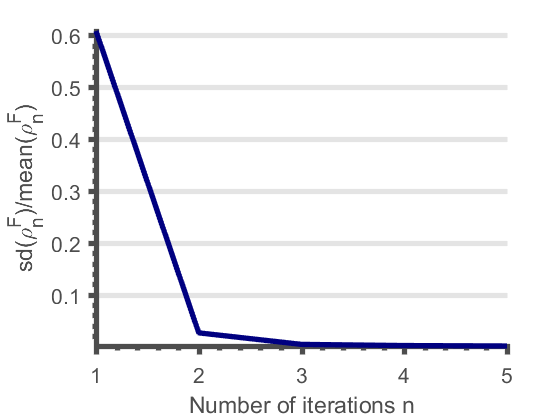}
\includegraphics[width=0.32\textwidth]{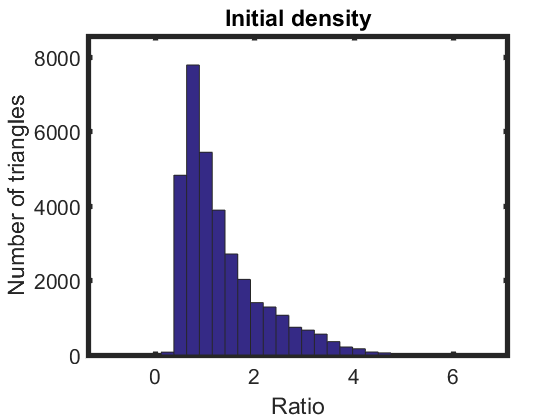}
\includegraphics[width=0.32\textwidth]{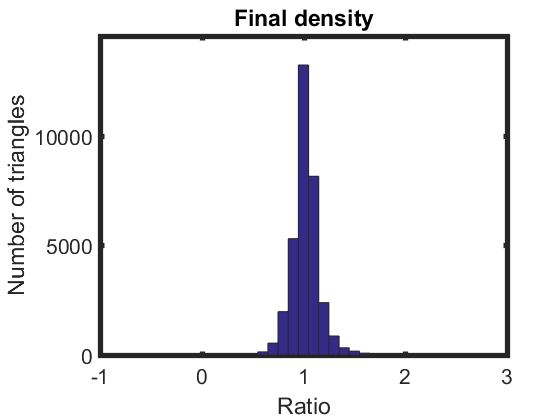}
\caption{Area-preserving parameterization of a lion face in $\R^3$. Top left: the initial shape. Top middle: the curvature-based locally authalic flattening initialization. Top right: the final density-equalizing map. Bottom left: the values of $\frac{\text{sd}(\rho^{\mathcal{F}}_n)}{\text{mean}(\rho^{\mathcal{F}}_n)}$. Bottom middle: the histogram of the density $\frac{\text{Initial area}}{\text{Initial flattened area}}$ on each flattened triangle element after the flattening initialization. Bottom right: the histogram of the density $\frac{\text{Initial area}}{\text{Final area}}$ on each triangle element of the final result. }
\label{fig:lion}
\end{figure}

Now consider computing the area-preserving mapping for a real surface mesh of a lion face in $\R^3$ using our algorithm. Again, we set the population as the area of each triangle element on the initial surface for achieving an area-preserving parameterization. Algorithm \ref{alg:chi_curvature} is used for the initialization step of our density-equalizing mapping algorithm. Figure \ref{fig:lion} shows the initial surface and the mapping result obtained by our density-equalizing mapping algorithm. For better visualization, we color the meshes with the mean curvature of the input lion face. The locally authalic initialization does not preserve the global area ratio; in particular, the nose of the lion is shrunk. By contrast, the final density-equalizing flattening map effectively preserves the area ratios.

\begin{figure}[t!]
\centering
\includegraphics[width=0.2\textwidth]{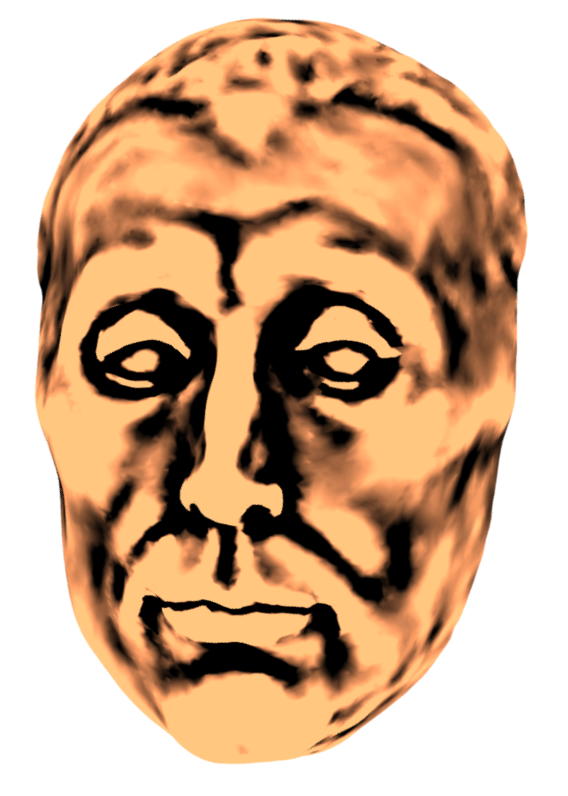}
\includegraphics[width=0.27\textwidth]{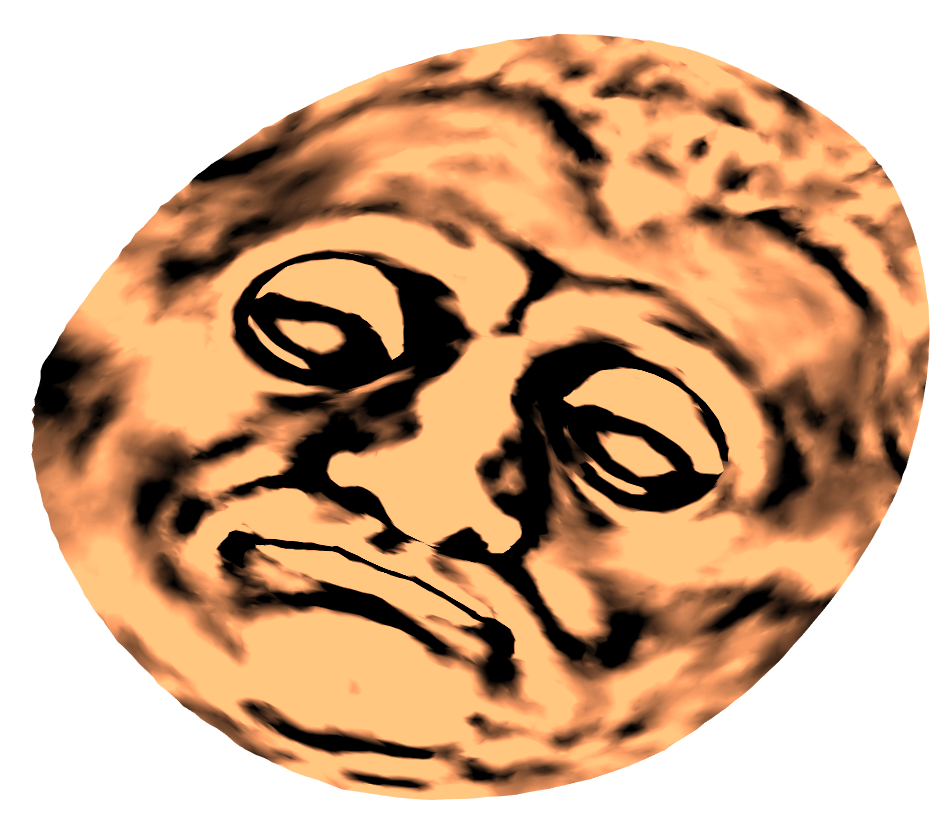}
\includegraphics[width=0.21\textwidth]{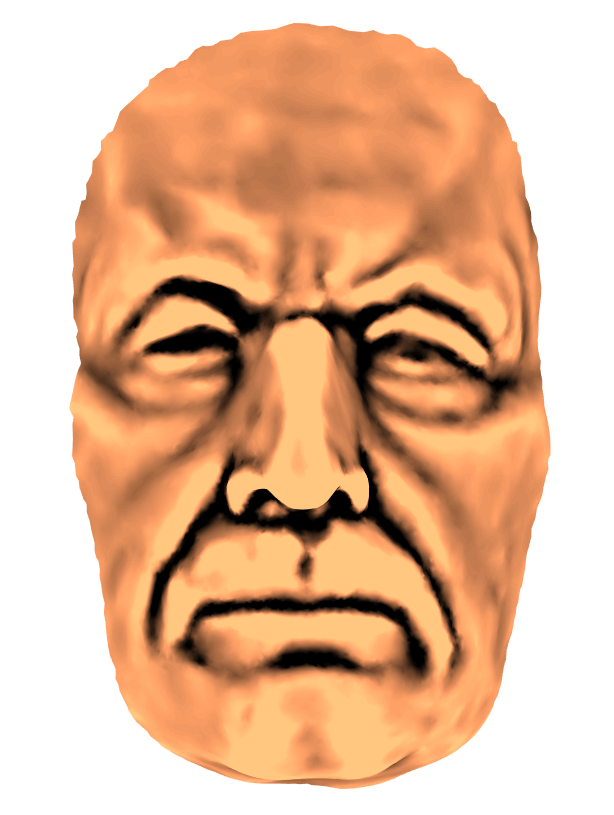}
\includegraphics[width=0.25\textwidth]{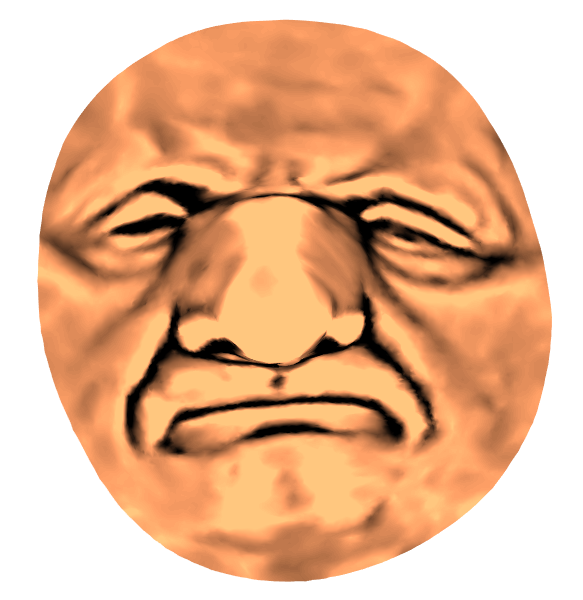}
\caption{Density-equalizing flattening maps with different effects obtained by our proposed algorithm. Left: the Niccol\`{o} da Uzzano model and the density-equalizing flattening map with the eyes magnified. Right: the Max Planck model and the density-equalizing flattening map with the nose magnified.}
\label{fig:effects}
\end{figure}

In addition, our algorithm can produce density-equalizing flattening maps with different effects by changing the input population. Figure \ref{fig:effects} shows two examples with different input populations. For the Niccol\`{o} da Uzzano model, we set the input population to be the area of each triangle element on the mesh except the eyes, and the population at the eyes to be 2 times the area of the triangles there. For the Max Planck model, we set the input population to be the area of each triangle element on the mesh except the nose, and the population at the nose to be 1.5 times the area of the triangles there. It can be observed that the resulting density-equalizing maps are respectively with the eyes and the nose magnified. 

\subsection{Numerical results of our algorithm}

\begin{table}[t]
\centering
\begin{tabular}{|C{17mm}|C{17mm}|C{15mm}|C{17mm}|C{15mm}|C{15mm}|}\hline
\textbf{Surface} & \textbf{No.~of triangles} & \textbf{Time (s)} & \textbf{No.~of iterations} & \textbf{Median of density} & \textbf{IQR of density} \\ \hline
Square & 10368 & 0.9858 & 5 & 1.0126 & 0.0799 \\ \hline 
Hexagon & 6144 & 0.4394 & 5 & 1.0227 & 0.0556 \\ \hline 
Gaussian & 10368 & 0.8667 & 4 & 1.0070 & 0.0307 \\ \hline 
Peaks & 4108 & 0.2323 & 4 & 1.0024 & 0.1325 \\ \hline 
Lion & 33369 & 2.4075 & 5 & 1.0220 & 0.1277 \\ \hline 
Niccol\`{o} da Uzzano & 25900 & 3.0077 & 8 & 1.0305 & 0.0823 \\ \hline
Max Planck & 26452 & 3.2395 & 11 & 1.0252 & 0.0633 \\ \hline
Human face & 6912 & 1.2356 & 6 & 1.0084 & 0.0571 \\ \hline 
US Map (Romney) & 46587 & 10.8280 & 3 & 1.0027 & 0.0146 \\ \hline 
US Map (Obama) & 46587 & 12.8330 & 4 & 0.9998 & 0.0147 \\ \hline 
US Map (Trump) & 46587 & 12.5154 & 4 & 1.0024 & 0.0176 \\ \hline 
US Map (Clinton) & 46587 & 12.8733 & 4 & 1.0003 & 0.0248 \\ \hline 
\end{tabular}
\caption{The performance of our algorithm. For each surface, we record the number of triangle elements, the time taken (in seconds) for the entire density-equalization algorithm (including the computation of initial map and the construction of sea), the number of iterations taken in the iterative scheme, and the median and interquartile range of the density defined on each triangle element by $\frac{\text{Given population}}{\text{Final area}}$.}
\label{table:result}
\end{table}

For a quantitative analysis, Table \ref{table:result} lists the detailed statistics of the performance of our algorithm on a number of simply-connected open meshes. From the time spent and the number of iterations needed, it can be observed that the convergence of our proposed algorithm is fast. Also, the median and the inter-quartile range of the density show that the density is well equalized under our algorithm. The experiments reflect the efficiency and accuracy of our proposed algorithm.

\begin{table}[t!]
\centering
\begin{tabular}{|C{40mm}|C{20mm}|C{20mm}|C{18mm}|}\hline
\textbf{Input population} & \textbf{Time by GN (s)} & \textbf{Time by our method (s)} & \textbf{Map difference} \\ \hline 
$5 + \frac{(x - \bar{x}) + (y - \bar{y})}{50}$ & 4.843 & 1.753 & 0.0009 \\ \hline 
$1 + e^{-\frac{(x - \bar{x})^2 + (y- \bar{y})^2}{1000}}$  & 4.452 & 1.501 & 0.0015 \\ \hline 
$2.5 + \sin \frac{\pi(x - \bar{x})}{25}$  & 4.959 & 1.776 & 0.0013 \\ \hline 
$1.5 + \sin \frac{\pi(x - \bar{x})}{25} \sin \frac{\pi(y - \bar{y})}{25}$  & 4.592 & 1.488 & 0.0026 \\ \hline 
\end{tabular}
\caption{Comparing the performance of our algorithm and GN deployed on a $100 \times 100$ square mesh with various input population functions. Here, the map difference is given by \text{mean}$\left(\frac{|z_{\text{prev}} - z_{\text{ours}}|}{\text{side length of square}}\right)$, where $z_{\text{prev}}$ and $z_{\text{ours}}$ are respectively the complex coordinates of the density-equalizing mapping results by GN and our method. $\bar x$ and $\bar y$ are the mean of the $x$-coordinates and the $y$-coordinates of the square.}
\label{table:comparison}
\end{table}

\begin{figure}[t!]
\centering
\includegraphics[width=0.24\textwidth]{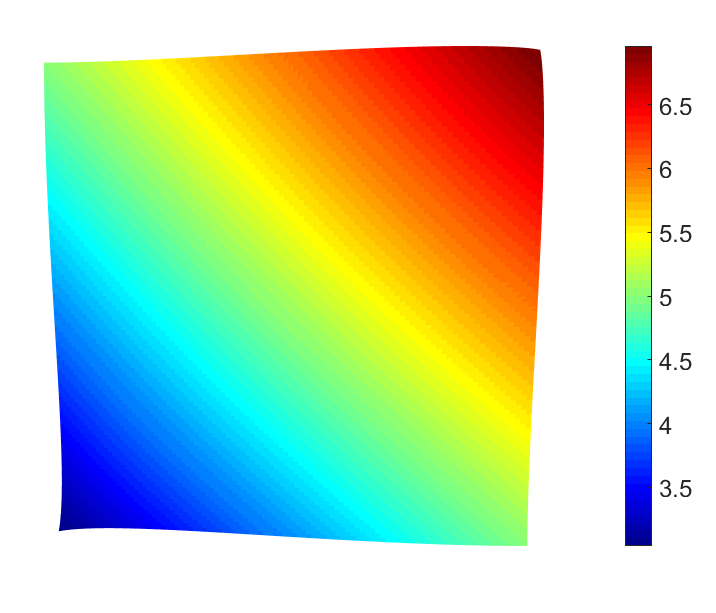}
\includegraphics[width=0.24\textwidth]{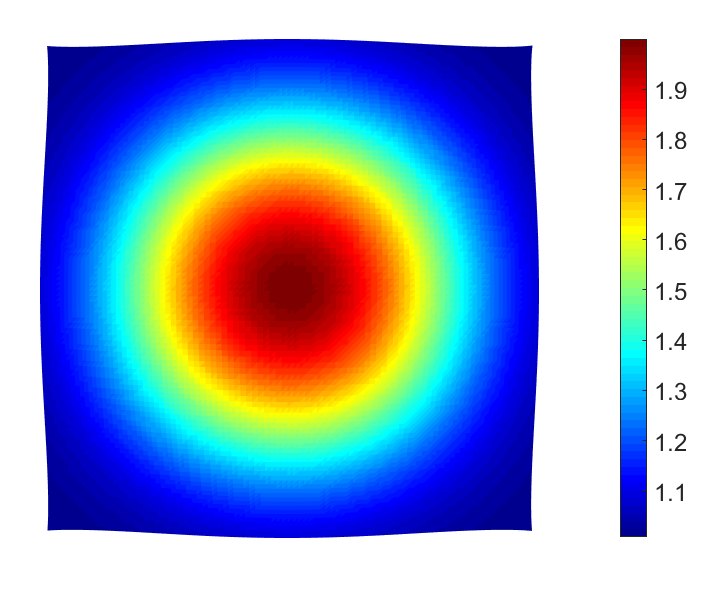}
\includegraphics[width=0.24\textwidth]{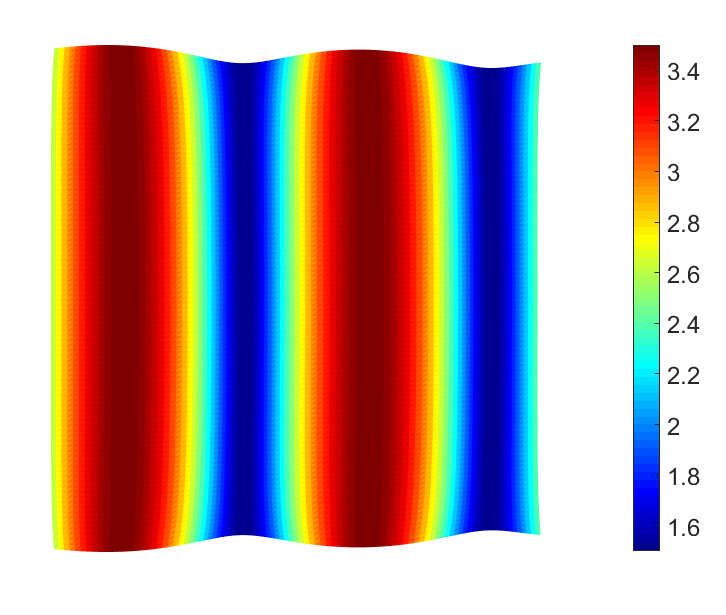}
\includegraphics[width=0.24\textwidth]{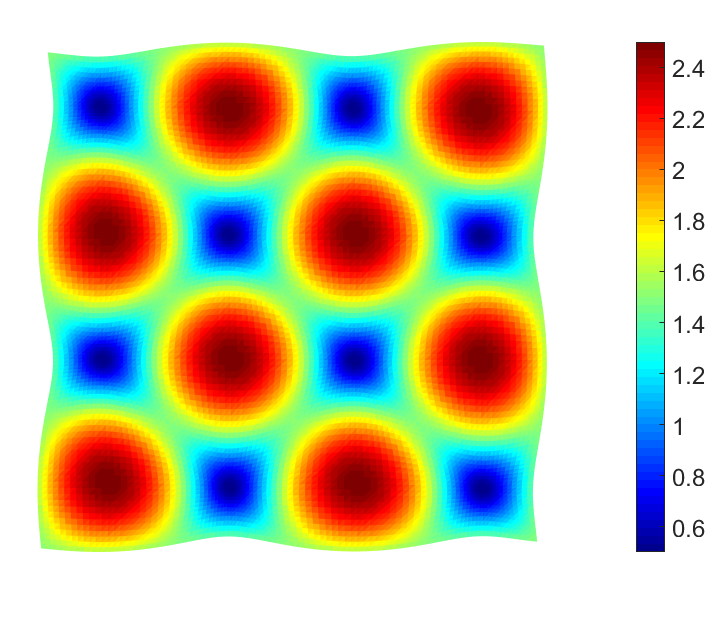}\\
\includegraphics[width=0.24\textwidth]{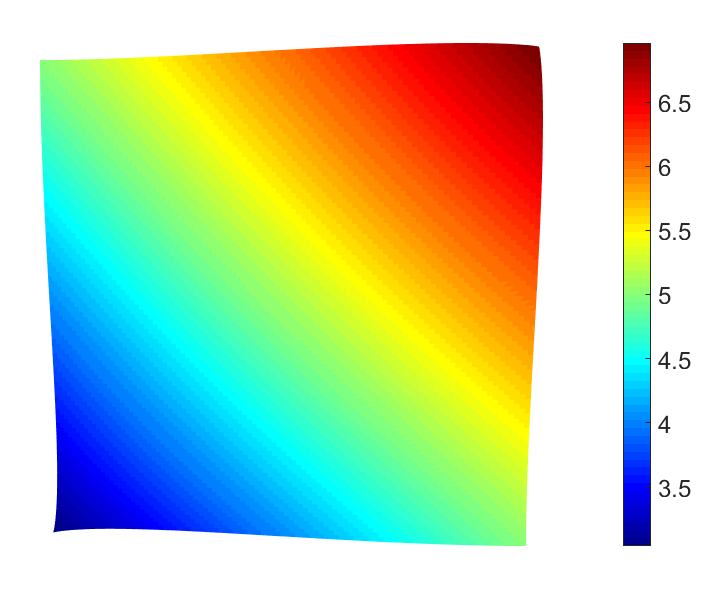}
\includegraphics[width=0.24\textwidth]{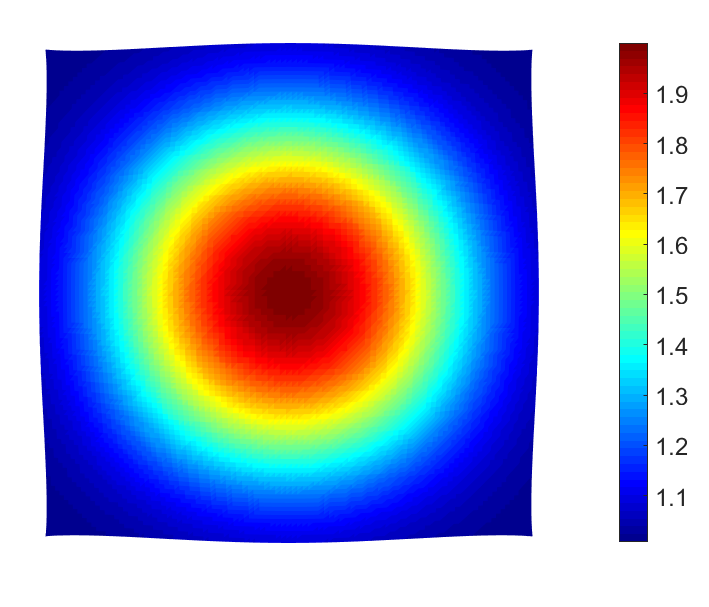}
\includegraphics[width=0.24\textwidth]{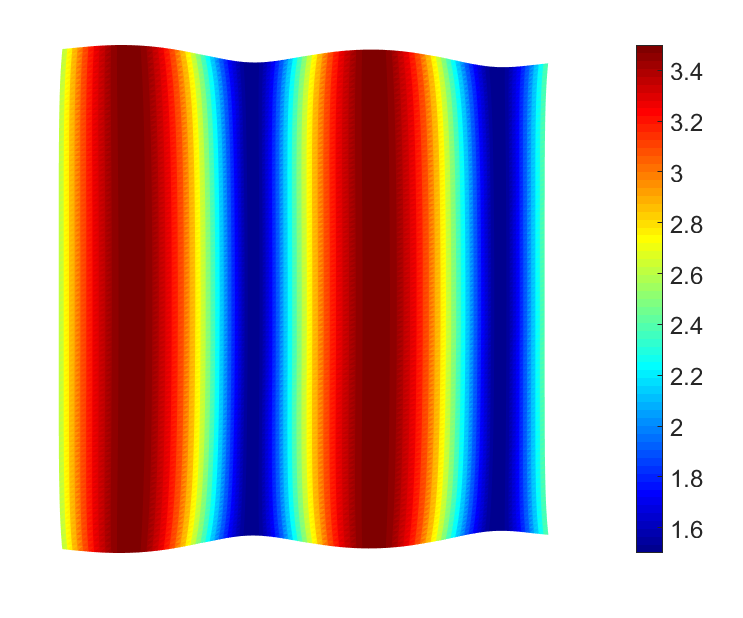}
\includegraphics[width=0.24\textwidth]{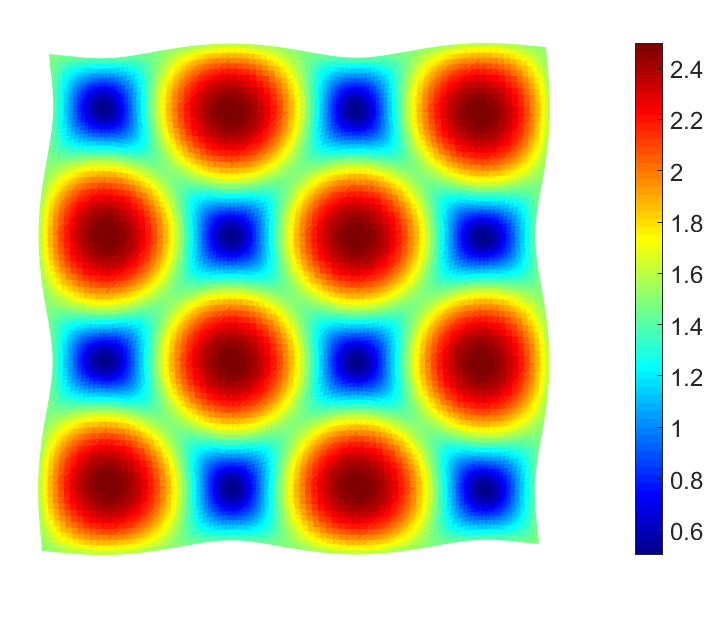}
\caption{The density-equalizing maps produced by our proposed algorithm and GN with various input population functions. Each column shows a set of experimental results color-coded by the input population function as described in Table \ref{table:comparison}. Top row: the results by GN. Bottom row: the results by our method. It can be observed that our method produces results as accurate as those by GN.}
\label{fig:square_comparison}
\end{figure}

We are also interested in analyzing the difference in the performance of our algorithm and GN with implementation available online~\cite{Cart}. Recall that GN works on finite difference grids. Therefore, for a fair comparison, we deploy the two methods on a $100 \times 100$ square grid $\{ (x,y) \in \mathbb{Z}^2: 0 \leq x, y \leq 99\}$ and compare the results. Following the suggestion by GN, the dimension of the sea is set to be two times the linear extent of the square grid in running GN. Various initial populations are tested for the computation of density-equalizing maps. Figure \ref{fig:square_comparison} shows several density-equalizing mapping results produced by the two methods. The statistics of the experiments are recorded in Table \ref{table:comparison}. With the accuracy well preserved, our method demonstrates an improvement on the computational time by over 60\% when compared to GN. 

We make a remark about the deformation of the sea under the density-equalizing process. Let $r$ be the displacement of every point at the sea from the origin before the deformation, and $\Delta r = r_{\text{final}} - r$ be the change in displacement of the point under the density-equalizing process. Figure \ref{fig:displacement_sea} shows several log--log plots of $\Delta r$ against $r$ outside the unit circle. We observe that $\Delta r$ and $r$ are related by the relationship $\Delta r \propto r^{-2}$ at the outer part of the sea. This suggests that setting a coarser sea at the outermost part does not affect the accuracy of the density-equalizing map.

\begin{figure}[t!]
 \centering
 \includegraphics[width=0.32\textwidth]{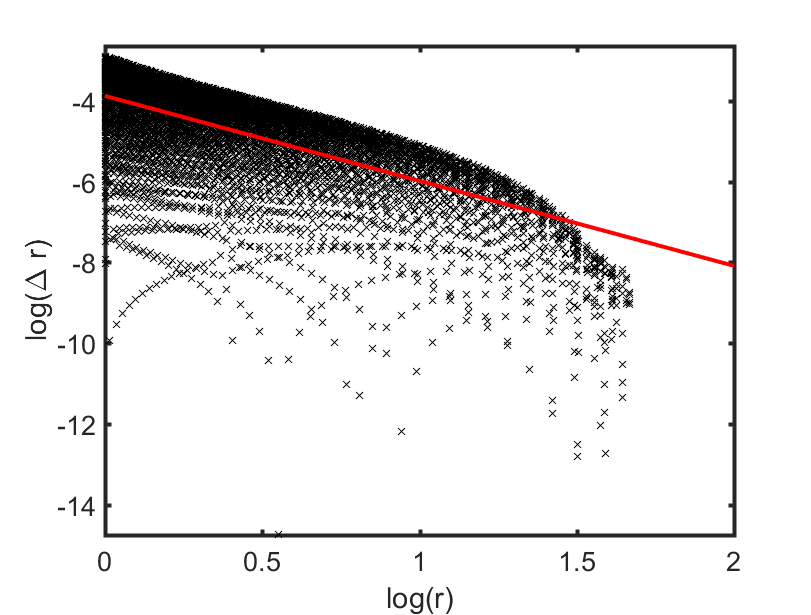}
 \includegraphics[width=0.32\textwidth]{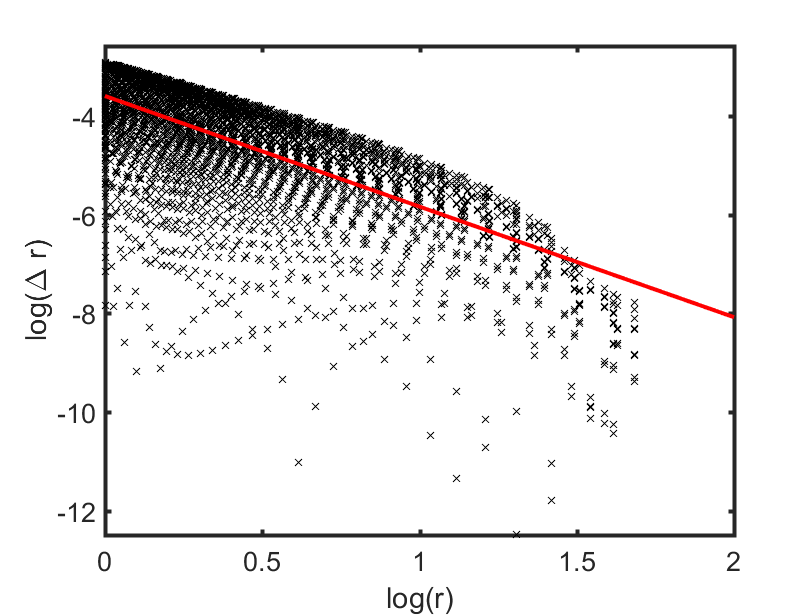}
 \includegraphics[width=0.32\textwidth]{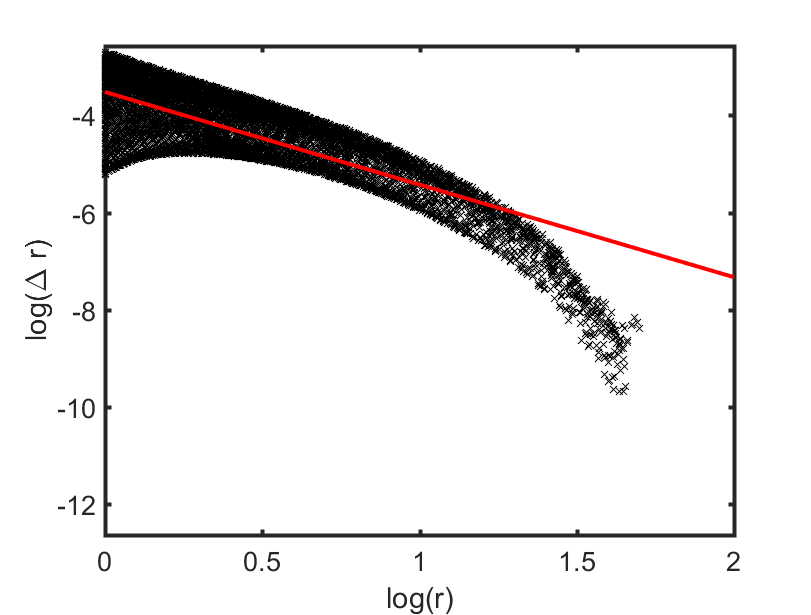}
 \caption{The log--log plot of the displacement of the sea under our density-equalizing algorithm. To study the effect at the outer sea, only the region outside the unit circle is considered. The $x$-axis represents the logarithm of the displacement of every point at the sea from the origin. The $y$-axis represents the logarithm of the change in the displacement under the density-equalizing map. Each cross represents a point at the sea, and the red line is the least-squares line. Left: the square example. Middle: the hexagon example. Right: the human face example.}
 \label{fig:displacement_sea}
\end{figure}

\section{Applications} \label{sect:applications}
Our proposed density-equalizing mapping algorithm is useful for various applications. In this section, we discuss two applications of our algorithm.

\subsection{Data visualization}

%

\begin{figure}[t!]
 \centering
 \includegraphics[width=0.49\textwidth]{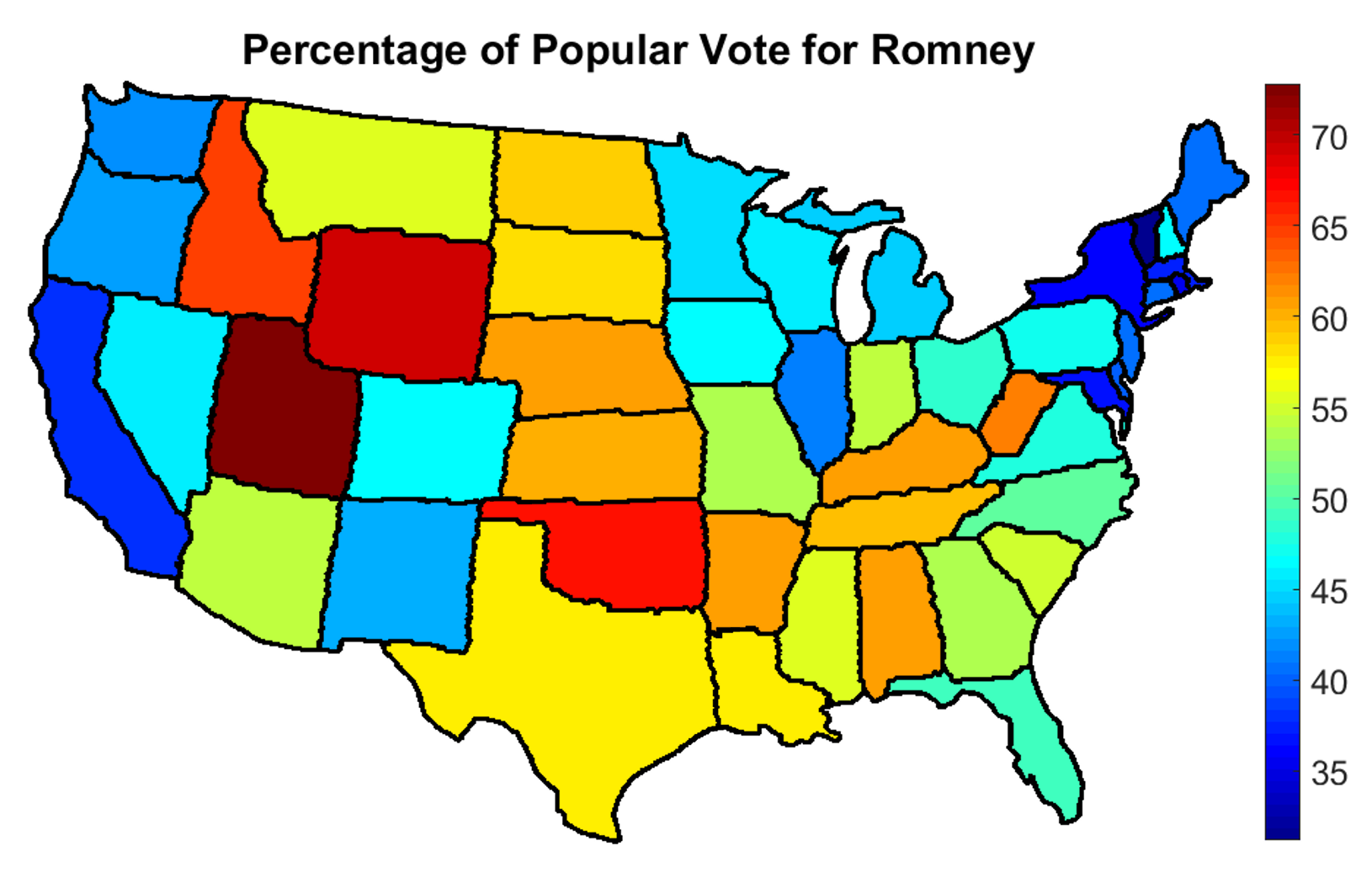}
 \includegraphics[width=0.49\textwidth]{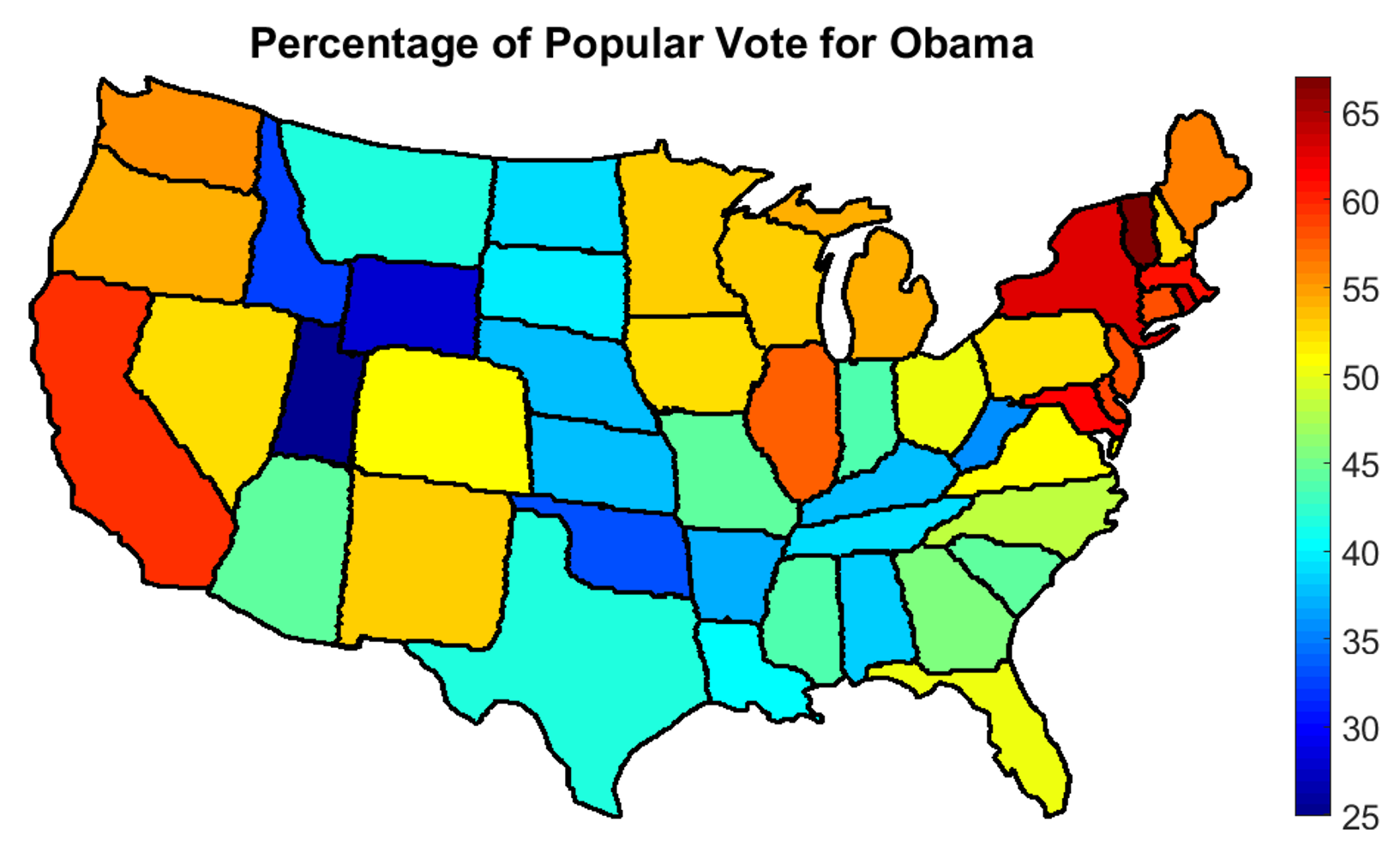}
 \includegraphics[width=0.49\textwidth]{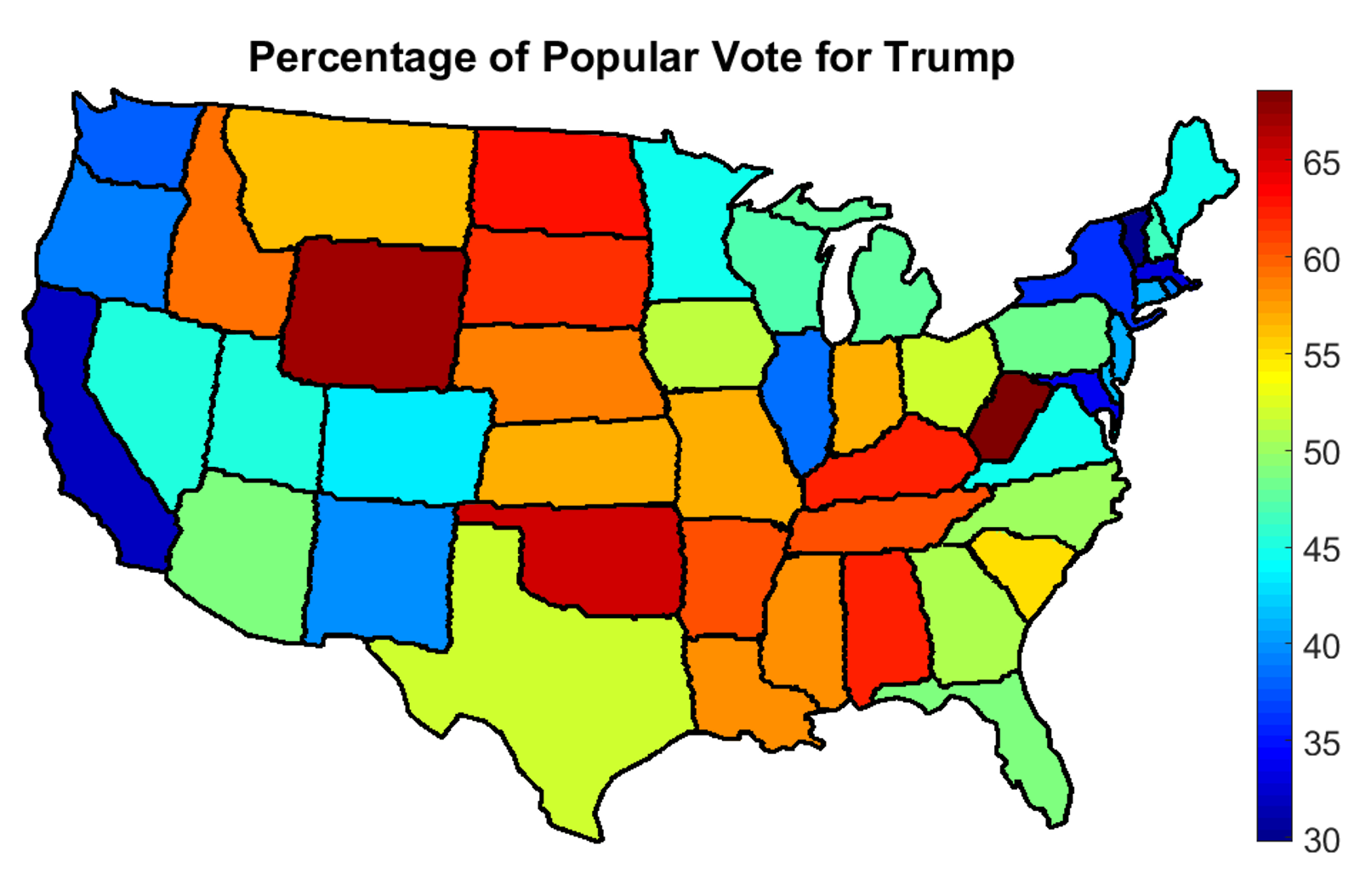}
 \includegraphics[width=0.49\textwidth]{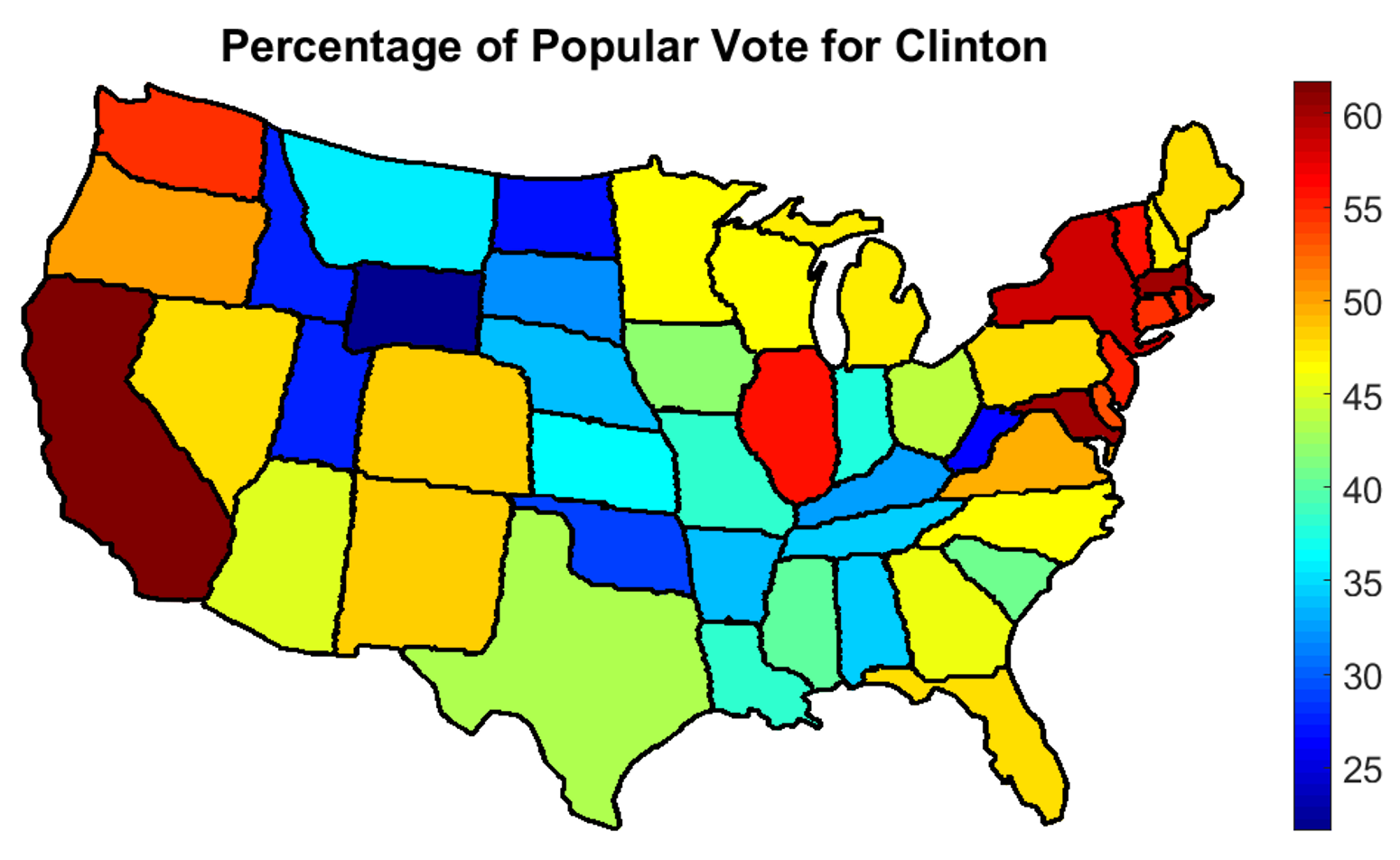}
 \caption{Percentage of popular vote in each state visualized on density-equalizing US maps (only including the contiguous 48 states). The triangulations are set to be transparent for enhancing the visual quality.}
 \label{fig:presidential}
\end{figure}

Similar to GN, our density-equalizing mapping algorithm can be used for data visualization. We consider visualizing the percentage of popular vote for the Republican party and the Democratic party in each state in the 2012 and 2016 US presidential elections. To visualize the data, we set the population on each state on a triangulated US map as the percentage of popular vote obtained by the two parties and run our proposed algorithm. Figure \ref{fig:presidential} shows the density-equalizing results. It can be observed that for the Republican party, the east coast and west coast are significantly shrunk. This reflects the relatively low percentage of popular vote obtained at those regions. By contrast, for the Democratic party, the east coast and west coast are significantly enlarged under the density-equalization, which reflects the relative high percentage of popular vote there. Some differences between the 2012 and the 2016 results can also be observed. For instance, the area of California becomes more extreme on the density-equalizing maps in 2016 when compared to those in 2012. For Trump, California has further shrunk on the map while for Clinton, it has further expanded on the map. Another example is West Virginia. It can be observed that the area of it has decreased in the map for Clinton when compared to that for Obama, while the area has increased in the map for Trump when compared to that for Romney. This example of US presidential election shows the usefulness of our density-equalizing mapping algorithm in data visualization.

\subsection{Adaptive surface remeshing}

\begin{figure}[t!]
 \centering
 \includegraphics[width=0.32\textwidth]{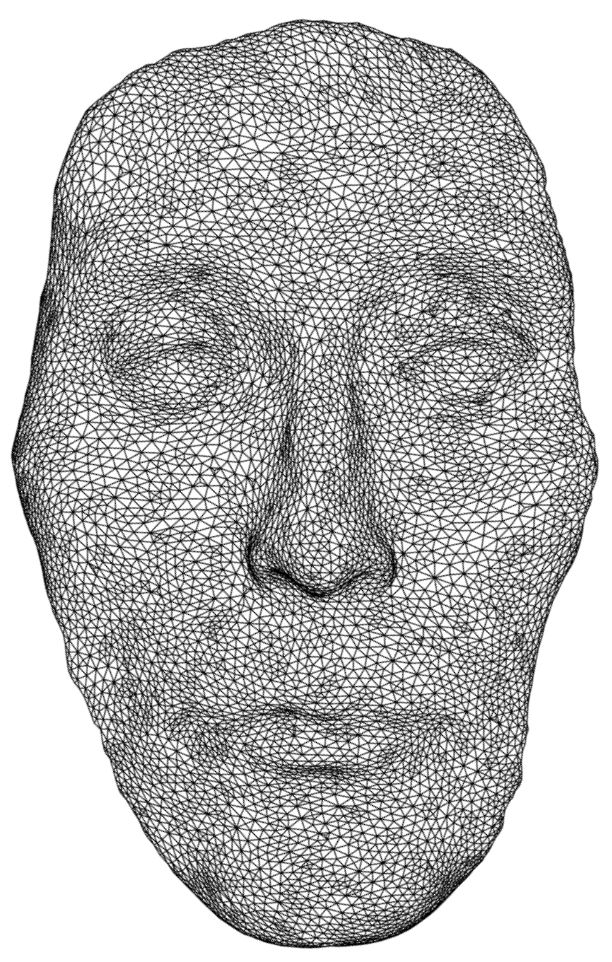}
 \includegraphics[width=0.32\textwidth]{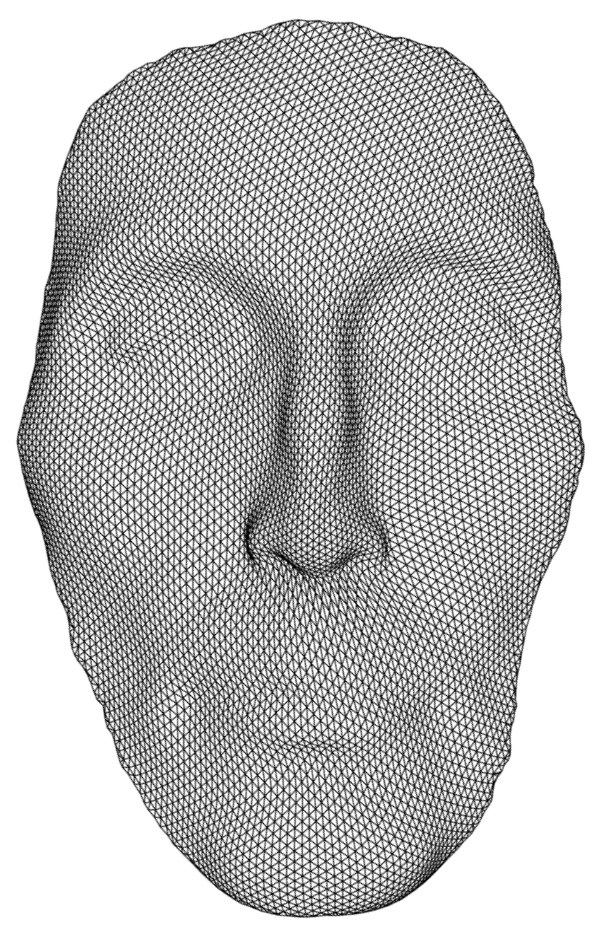}
 \includegraphics[width=0.32\textwidth]{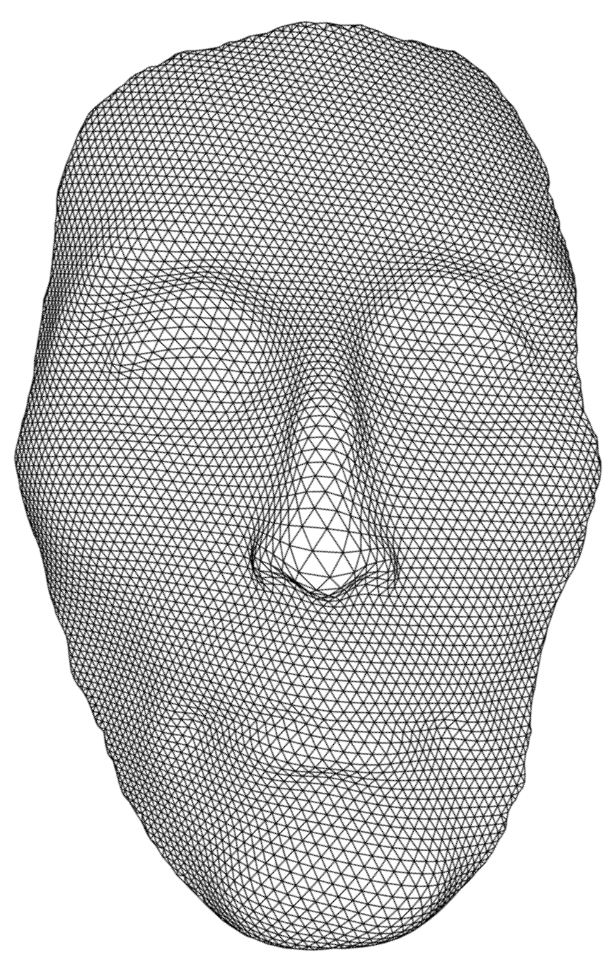}
 \includegraphics[width=0.42\textwidth]{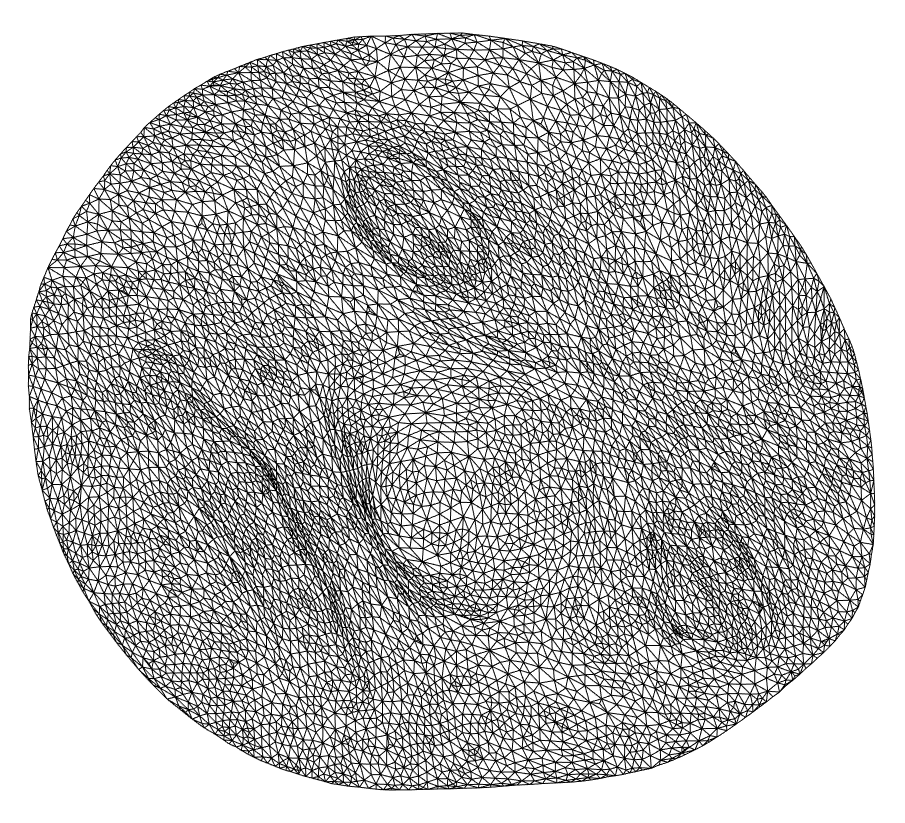}
 \includegraphics[width=0.4\textwidth]{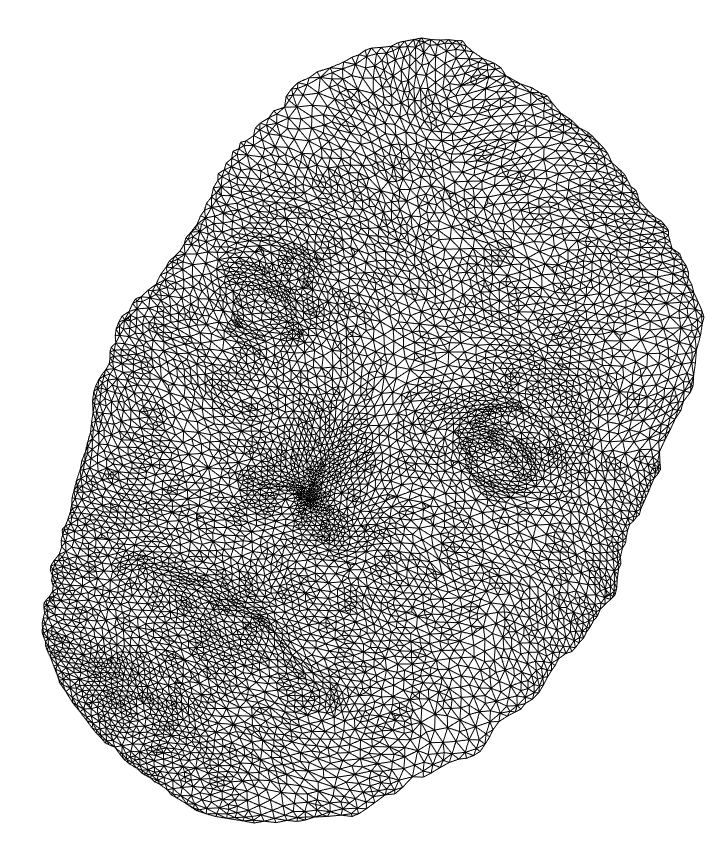}
 \caption{Remeshing a human face. Top left: the original human face. Top middle: the remeshing result via our parameterization. Top right: the remeshing result via the free-boundary conformal parameterization by Desbrun et al.~\cite{Desbrun02}. Bottom left: the density-equalizing parameterization by our algorithm. Bottom right: the free-boundary conformal parameterization by Desbrun et al.~\cite{Desbrun02}.}
 \label{fig:remeshing}
\end{figure}

Note that the input population affects the size of different regions in the resulting density-equalizing map. Specifically, a higher population leads to a magnification and a lower population leads to a shrinkage. Using this property of the density-equalizing map, we can perform adaptive surface remeshing easily.

Let $S$ be a surface to be remeshed. Given a population, we first compute the density-equalizing map $f:S \to \mathbb{C}$. Now consider a set of uniformly distributed points $\mathcal{P}$ on the density-equalizing map. We triangulate the set of points and denote the triangulation by $\mathcal{T}$. Then, using the inverse mapping $f^{-1}$, we can interpolate $\mathcal{P}$ onto $S$. The mesh $(f^{-1}(\mathcal{P}),\mathcal{T})$ gives a remeshed representation of the surface $S$.

Now, to increase the level of details at a region of $S$, we can set a larger population at there in running our density-equalizing mapping algorithm. Since the region is enlarged in the mapping result and $\mathcal{P}$ is uniformly distributed, more points will lie on that part and hence the inverse mapping will map more points back onto that particular region of $S$. This completes our adaptive surface remeshing scheme.

Figure \ref{fig:remeshing} shows an example of remeshing a triangulated human face using the abovementioned scheme. We set the population to be the triangle area of the original mesh in running our algorithm. To highlight the advantage of the use of our density-equalizing map, we compare the remeshing result with that obtained via a conventional free-boundary conformal parameterization method~\cite{Desbrun02}. The eyes and the nose of the human face are enlarged in our density-equalizing mapping result, while such features are shrunk in the conformal parameterization because of the preservation of conformality. This difference causes significantly different remeshing results. Also, note that the representation of the nose is poor in the remeshing result via conformal parameterization. By contrast, the remeshing result via our density-equalizing mapping algorithm is with a more balanced distribution of points. This example demonstrates the strength of our algorithm in surface remeshing.

\section{Discussion} \label{sect:discussion}
In this work, we have proposed an efficient algorithm for computing density-equalizing flattening maps of simply-connected open surfaces in $\R^3$. When compared to GN, our method is particularly well suited to planar domains with complex geometry because of the use of triangular meshes. With this advantage, our method can possibly lead to a wider range of applications of density-equalizing maps in data visualization. Our method is also well suited for handling disk-like surfaces in $\mathbb{R}^3$ such as human faces. This suggests a new approach for adaptive surface remeshing via density-equalizing maps. When compared to the existing parameterization-based remeshing approaches, our method can easily control the remeshing quality at different regions of the surfaces by changing the population at those regions.

Since the density diffusion process is solved on a triangular mesh, the triangle quality affects the accuracy of the discretization and hence the final density-equalization result. If the triangular mesh consists of highly irregular triangle elements, the ultimate density distribution may not be optimal even after the algorithm converges. Also, since the discretization is based on the triangles for every step in our algorithm, if the input population is too extreme or highly discontinuous, the triangles may become highly irregular at a certain step and affects the accuracy of the subsequent results. In other words, triangle meshes with moderate triangle quality and input population are desired. Besides, for surfaces in $\R^3$ with a highly tubular shape and with the boundary lying at one end, the flattening step may causes extremely squeezed regions on the planar domain. In this case, the accuracy of the subsequent computations for density equalization may be affected.

Our current work primarily focuses on simply-connected surfaces in $\R^3$, but it can be naturally extended to general surfaces. For instance, density-equalizing maps of multiply-connected surfaces can be computed by filling up the holes and treating them as the sea in our proposed algorithm. Similarly, density-equalizing maps of multiple disconnected surfaces can be handled with the aid of a large sea. 

\section*{Acknowledgments}
This work was partially supported by a fellowship from the Croucher Foundation (GPT Choi).

\section*{Appendix}
We prove that the curvature-based curve flattening step in Sec. \ref{sect:step1} produces a simple closed convex curve.
\begin{proposition}
Let $\varphi:[0,l_{\gamma}] \to \R^2$ be the arclength parameterized curve defined as in Sec. \ref{sect:step1}. Consider the new curve $\varPhi:[0,l_{\gamma}] \to \R^2$ defined by 
\begin{equation}
\varPhi(s) =  \varphi(s) - \frac{s}{l_{\gamma}}\left(\varphi(l_{\gamma}) - \varphi(0)\right).
\end{equation}
$\varPhi$ is a simple closed convex curve.
\end{proposition}

\begin{proof}
It is easy to note that $\varPhi(0) = \varphi(0) = \varPhi(l_{\gamma})$ and hence $\varPhi$ is closed. Since $\varphi$ is an arclength parameterized curve, for any $0 \leq a < b \leq l_{\gamma}$, we have
\begin{equation}\label{eqt:straight_line}
\|\varphi(b)-\varphi(a)\| \leq \int_{a}^{b} \|\varphi'(s)\| ds = b - a,
\end{equation}
where the equality holds if and only if $\varphi([a,b])$ is a straight line. In particular, since $\gamma$ is the boundary of the original simply-connected open surface, by our construction of $\varphi$, we have $\|\varphi(l_{\gamma})-\varphi(0)\| \ll l_{\gamma}.$

We now prove that the signed curvature of $\varPhi$, denoted by $k_{\varPhi}$, is non-negative for all $s \in [0,l_{\gamma}]$. Denote $\varphi(s) = (x(s),y(s))$ and $\varPhi(s) = (X(s), Y(s))$. We have
\begin{equation}
\begin{split}
\varPhi' & = (X', Y') \\
& = \left(x'(s) - \frac{1}{l_{\gamma}}\left(x(l_{\gamma}) - x(0)\right), y'(s) - \frac{1}{l_{\gamma}}\left(y(l_{\gamma}) - y(0)\right)\right)\\
& = \left(\cos \theta(s) - \frac{1}{l_{\gamma}}\left(x(l_{\gamma}) - x(0)\right), \sin \theta(s) - \frac{1}{l_{\gamma}}\left(y(l_{\gamma}) - y(0)\right)\right)
\end{split}
\end{equation} 
and
\begin{equation}
\varPhi'' = (X'', Y'') = (x'', y'') = \left(- k_{\varphi}(s) \sin \theta(s) , k_{\varphi}(s) \cos \theta(s)\right).
\end{equation} 
Hence, we have
\begin{equation} \label{eqt:curvature_relation}
X'Y''-X''Y' = k_{\varphi}(s) \left(1 - \frac{\left(x(l_{\gamma}) - x(0)\right) \cos \theta(s) - \left(y(l_{\gamma}) - y(0)\right) \sin \theta(s)}{l_{\gamma}}\right).
\end{equation} 
Now recall that by \eqref{eqt:nonnegative_curvature_gamma}, $k_{\varphi}(s) \geq 0$ for all $s$. Also, we have

\begin{equation} \label{eqt:inequality}
\begin{split}
&\left(x(l_{\gamma}) - x(0)\right) \cos \theta(s) - \left(y(l_{\gamma}) - y(0)\right) \sin \theta(s) \\
\leq & \sqrt{\left(x(l_{\gamma}) - x(0)\right)^2 +  \left(y(l_{\gamma}) - y(0)\right)^2} \sqrt{\cos^2 \theta(s) + \sin^2 \theta(s)}\\
= & \sqrt{\left(x(l_{\gamma}) - x(0)\right)^2 +  \left(y(l_{\gamma}) - y(0)\right)^2}\\
= & \left\|\varphi(l_{\gamma}) - \varphi(0)\right\|\\
\leq & l_{\gamma}.
\end{split}
\end{equation} 
Here, the first inequality follows from the Cauchy--Schwarz inequality, and the second inequality follows from \eqref{eqt:straight_line}. Therefore, we have
\begin{equation}
1 - \frac{\left(x(l_{\gamma}) - x(0)\right) \cos \theta(s) - \left(y(l_{\gamma}) - y(0)\right) \sin \theta(s)}{l_{\gamma}} \geq 0 \text{ for all } s \in [0,l_{\gamma}],
\end{equation}
and it follows that
\begin{equation}\label{eqt:nonnegative_curvature}
k_{\varPhi}(s) = \frac{X'Y''-X''Y'}{(X'^2+Y'^2)^{3/2}} \geq 0 \text{ for all } s \in [0,l_{\gamma}].
\end{equation} 

We proceed to show that $\varPhi$ is simple. Note that since $\varPhi$ is a closed plane curve, the total curvature of $\varPhi$ should satisfy
\begin{equation}
\int_0^{l_{\varphi}} k_{\varPhi}(s) ds  = 2\pi n_{\varPhi},
\end{equation}
where $n_{\varPhi}$ is the turning number of $\varPhi$. From the above results, we have
\begin{equation}
\begin{split}
2\pi n_{\varPhi} &= \int_0^{l_{\varphi}} \frac{k_{\varphi}(s) \left(1 - \frac{\left(x(l_{\gamma}) - x(0)\right) \cos \theta(s) - \left(y(l_{\gamma}) - y(0)\right) \sin \theta(s)}{l_{\gamma}}\right)}{(X'^2+Y'^2)^{3/2}} ds\\
&\leq \left(\int_0^{l_{\varphi}} k_{\varphi}(s) ds \right) \max_{s \in [0, l_{\gamma}]} \left|\frac{\left(1 - \frac{\left(x(l_{\gamma}) - x(0)\right) \cos \theta(s) - \left(y(l_{\gamma}) - y(0)\right) \sin \theta(s)}{l_{\gamma}}\right)}{(X'^2+Y'^2)^{3/2}}\right|\\
&= 2\pi \max_{s \in [0, l_{\gamma}]} \left( \frac{1 - \frac{\left(x(l_{\gamma}) - x(0)\right) \cos \theta - \left(y(l_{\gamma}) - y(0)\right) \sin \theta}{l_{\gamma}}}{\left(1+\frac{(x(l_{\gamma})-x(0))^2+(y(l_{\gamma})-y(0))^2}{l_{\gamma}} - \frac{2\left( \left(x(l_{\gamma}) - x(0)\right) \cos \theta + \left(y(l_{\gamma}) - y(0)\right) \sin \theta \right)}{l_{\gamma}}\right)^{3/2}}\right).
\end{split}
\end{equation}
Substituting $A = \frac{x(l_{\gamma}) - x(0)}{l_{\gamma}}$ and $B = \frac{y(l_{\gamma}) - y(0)}{l_{\gamma}}$, the above equation becomes
\begin{equation}
\begin{split}
&2\pi \max_{s \in [0, l_{\gamma}]} \left( \frac{1 - A \cos \theta(s) + B \sin \theta(s)}{\left(1+A^2+B^2 - 2A\cos \theta(s) - 2B\sin \theta(s)\right)^{3/2}}\right)\\
= & 2\pi \max_{s \in [0, l_{\gamma}]} \left( \frac{1 - C \cos (\theta(s)+\eta)}{\left(1+ C^2 - 2C\cos(\theta(s)-\eta))\right)^{3/2}}\right),
\end{split}
\end{equation}
where $C = \sqrt{A^2+B^2} = \frac{\left\|\varphi(l_{\gamma}) - \varphi(0)\right\|}{l_{\gamma}} \ll 1$ and $\eta = \tan^{-1}\frac{B}{A} = \tan^{-1} \frac{y(l_{\gamma}) - y(0)}{x(l_{\gamma}) - x(0)}$. Hence, it is easy to see that
\begin{equation}
\max_{s \in [0, l_{\gamma}]} \left( \frac{1 - C \cos (\theta(s)+\eta)}{\left(1+ C^2 - 2C\cos(\theta(s)-\eta))\right)^{3/2}}\right) \leq \frac{1 + C}{\left(1-C\right)^3} < 2.
\end{equation}
This implies that 
\begin{equation}
2\pi n_{\varPhi} < 4 \pi \Rightarrow n_{\varPhi} < 2.
\end{equation}
It follows that $n_{\varPhi} = 1$ and hence $\varPhi$ is simple. Finally, note that a simple closed curve is convex if and only if its signed curvature does not change sign~\cite{Gray98}. From \eqref{eqt:nonnegative_curvature}, we conclude that $\varPhi$ is a simple closed convex curve.
\end{proof}
\end{document}